\documentclass[runningheads]{llncs}
\usepackage[T1]{fontenc}

\usepackage{times} 

\usepackage{amsmath}
\usepackage{amsfonts}
\usepackage{amssymb}
\usepackage{graphicx}
\usepackage{color}
\usepackage{xcolor}
\usepackage{colortbl}
\usepackage{url}

\usepackage[linesnumbered,ruled,vlined]{algorithm2e}

\SetCommentSty{mycommfont}

\usepackage{paralist} 
\usepackage{listings} 
\usepackage{bbding} 
\usepackage{multirow}

\usepackage{bbm}
\usepackage{tabularx, booktabs}
\newcolumntype{C}[1]{>{\centering\let\newline\\\arraybackslash\hspace{0pt}}m{#1}}
\newcolumntype{R}[1]{>{\raggedleft\let\newline\\\arraybackslash\hspace{0pt}}m{#1}}
\newcolumntype{L}[1]{>{\raggedright\let\newline\\\arraybackslash\hspace{0pt}}m{#1}}

\usepackage[title]{appendix}
\usepackage{enumitem}

\usepackage{graphicx}
\usepackage{color}
\usepackage{diagbox}
\usepackage{subfig}
\usepackage{hyperref}
\usepackage{wrapfig}
\usepackage{bigdelim}

\makeatletter
\newcommand\footnoteref[1]{\protected@xdef\@thefnmark{\ref{#1}}\@footnotemark}
\makeatother

\newtheorem{prop}{Proposition}
\newtheorem{h0}{Null Hypothesis}

\newcommand{\hide}[1]{}

\newcommand{\bit}{\begin{compactitem}}
\newcommand{\eit}{\end{compactitem}}
\newcommand{\ben}{\begin{compactenum}}
\newcommand{\een}{\end{compactenum}}

\newcommand{\method}{\textsc{NetEffect}\xspace}
\newcommand{\methodtest}{\textsc{NetEffect\textunderscore Test}\xspace}
\newcommand{\methodest}{\textsc{NetEffect\textunderscore Est}\xspace}
\newcommand{\methodexp}{\textsc{NetEffect\textunderscore Exp}\xspace}
\newcommand{\methodhom}{\textsc{NetEffect}-Hom\xspace}

\newcommand{\fabp}{\textsc{FaBP}\xspace}
\newcommand{\linbp}{\textsc{LinBP}\xspace}
\newcommand{\hols}{\textsc{HOLS}\xspace}

\newcommand{\gcn}{\textsc{GCN}\xspace}

\newcommand{\mixhop}{\textsc{MixHop}\xspace}

\newcommand{\gprgnn}{\textsc{GPR-GNN}\xspace}
\newcommand{\appnp}{\textsc{APPNP}\xspace}

\newcommand{\explain}{{Explainable\xspace}}

\newcommand{\scale}{{Scalable\xspace}}

\newcommand{\neteffect}{{\em generalized network-effects}\xspace}
\newcommand{\accurate}{Accurate\xspace}
\newcommand{\general}{General\xspace}

\newcommand{\nef}{GNE\xspace}

\newcommand{\emphasis}{``emphasis'' matrix\xspace}

\newcommand{\NEF}{Network Effect Formula\xspace}
\newcommand{\xophily}{\textsc{x-}ophily\xspace}

\newcommand{\theory}{{Principled\xspace}}
\newcommand{\codeurl}{\url{https://github.com/mengchillee/NetEffect}}

\newcommand{\rulesep}{\unskip\ \vrule\ }

\newcommand{\classes}{c\xspace}
\newcommand{\n}{n\xspace}

\newcommand{\red}[1]{\textcolor{red}{#1}}
\newcommand{\blue}[1]{\textcolor{blue}{#1}}

\newcommand{\gold}[1]{\cellcolor{green}{#1}}
\newcommand{\silver}[1]{\cellcolor{green!35}{#1}}
\newcommand{\bronze}[1]{\cellcolor{green!12}{#1}}

\newcommand{\goldc}[1]{\colorbox{green}{#1}}
\newcommand{\silverc}[1]{\colorbox{green!35}{#1}}
\newcommand{\bronzec}[1]{\colorbox{green!12}{#1}}

\newcommand{\RNum}[1]{\uppercase\expandafter{\romannumeral #1\relax}}

\newcommand{\emphasize}[1]{\textbf{\underline{\smash{#1}}}}

\begin{document}

\title{\method: Discovery and Exploitation of Generalized Network Effects}



\author{Meng-Chieh Lee\inst{1} \and
Shubhranshu Shekhar\inst{2} \and
Jaemin Yoo\inst{3} \and
Christos Faloutsos\inst{1}}
\institute{Carnegie Mellon University, Pittsburgh, USA  \\
\email{\{mengchil, christos\}@cs.cmu.edu}
\and Brandeis University, Waltham, USA \\
\email{sshekhar@brandeis.edu} \\
\and KAIST, Seoul, South Korea \\
\email{jaemin@kaist.ac.kr}
}


\maketitle

\vspace{-4mm}
\begin{abstract}
Given a large graph with few node labels, how can we 
(a) identify whether there is \neteffect~(\nef) or not, 
(b) estimate \nef to explain the interrelations among node classes, and 
(c) exploit \nef efficiently to improve the performance on downstream tasks? 
The knowledge of \nef is valuable for various tasks like node classification and targeted advertising. 
However, identifying \nef such as homophily, heterophily or their combination is challenging in real-world graphs due to limited availability of node labels and noisy edges. 
We propose \method, a graph mining approach to address the above issues, enjoying the following properties:
(i) \theory: a statistical test to determine the presence of \nef in a graph with few node labels; 
(ii) \general and \explain: a closed-form solution to estimate the specific type of \nef observed; and 
(iii) \accurate and \scale: the integration of \nef for accurate and fast node classification.
Applied on real-world graphs, \method discovers the unexpected absence of \nef in numerous graphs, which were recognized to exhibit heterophily. 
Further, we show that incorporating \nef is effective on node classification.
On a million-scale real-world graph,  
\method achieves {\bf over 7$\mathbf{\times}$} speedup ({\em $14$ minutes} vs. $2$ hours) compared to most competitors.

\end{abstract}
\vspace{-8mm}
\vspace{-3mm}

\section{Introduction} \label{sec:intro}
\vspace{-3mm}
Given a large graph with few node labels and no node features, how to check whether the graph structure is useful for classifying nodes or not? 
Node classification is often employed to infer labels on large real-world graphs.
Since manual labeling is expensive and time-consuming, it is common that only few node labels are available.
For example, in a million-scale social network, identifying even a fraction (say $5\%$) of users' groups is prohibitive, limiting the application of methods that assume many labels are given.
Recently, with prevalence of graphs in industry and academia alike, there is a growing need among users to know whether these graph structures provide meaningful information for inference tasks.
Therefore, before investing a huge amount of time and resources into potentially unsuccessful experiments, a preliminary test is earnestly needed.\looseness=-1

\begin{figure}[t]
\centering
\subfloat[\label{fig:c1} \methodtest:\\\theory]
{\includegraphics[scale=0.3]{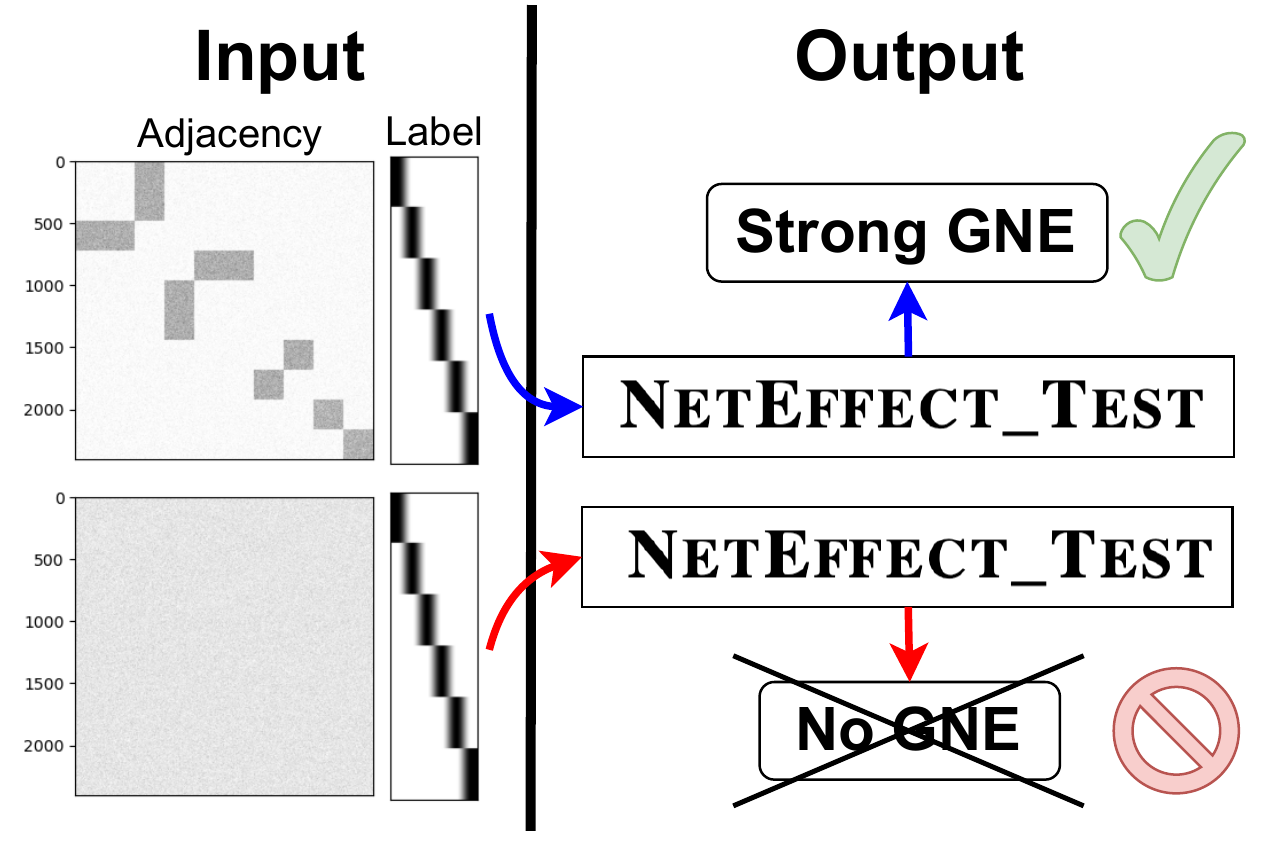}}
\rulesep
\subfloat[\label{fig:c2} \methodest:\\\explain\xspace and \general]
{\includegraphics[scale=0.3]{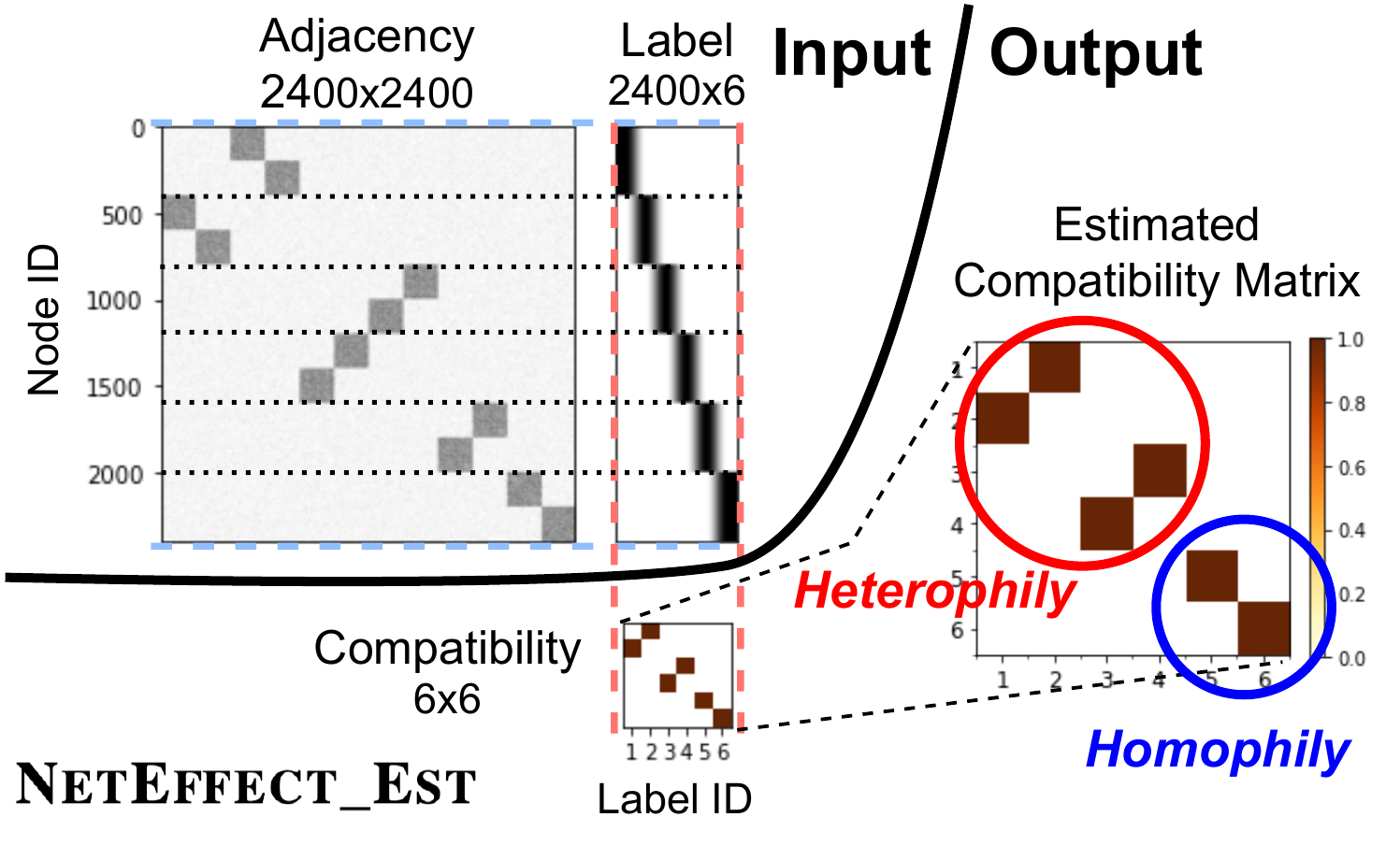}}
\rulesep
\subfloat[\label{fig:c3} \methodexp:\\\accurate and \scale]
{\includegraphics[scale=0.3]{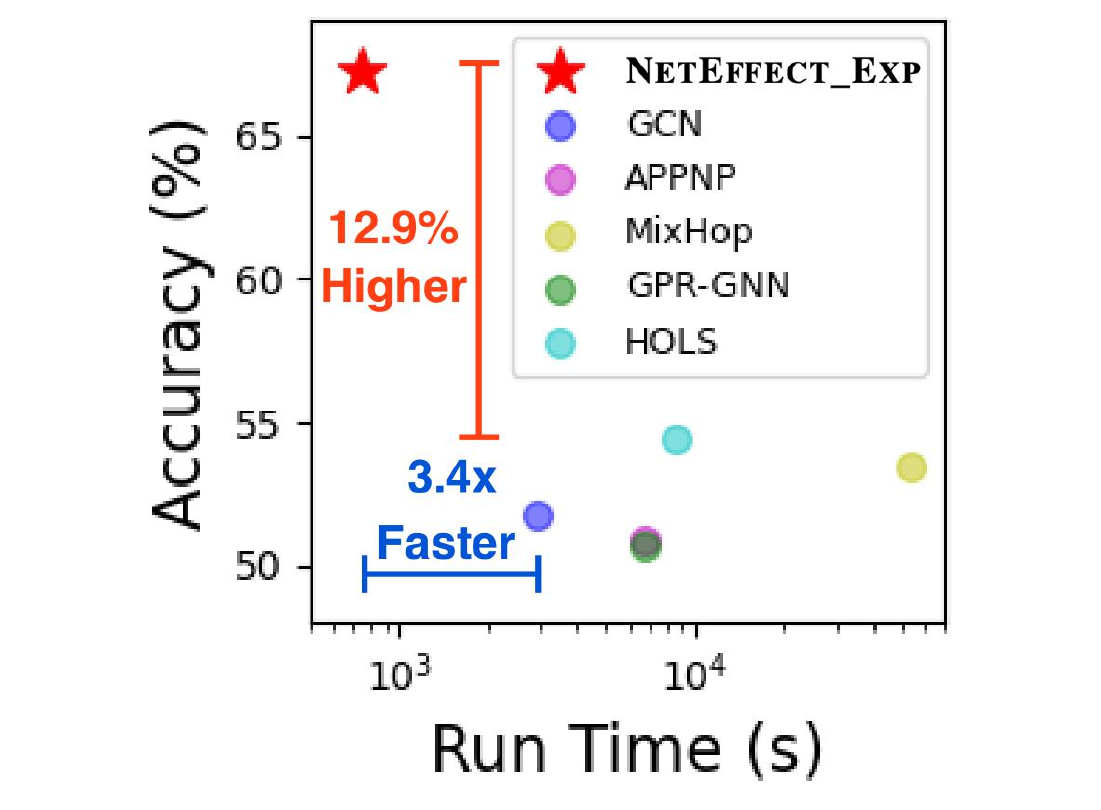}}
\vspace{-3mm}
\caption{\label{fig:crown} \emphasize{\method works well}, thanks to its three novel contributions:
(a) \methodtest statistically the existence of \nef.
(b) \methodest explains the graph with 
the \xophily compatibility matrix.
(c) \methodexp wins and is fast.
}
\vspace{-6mm}
\end{figure}


That is to say, we want to know whether the given graph has \neteffect (\nef) or not.
A graph with \nef provides meaningful information through the structure that can be used to identify the labels of nodes.
For example, ``talkative person tends to make friends with talkative ones'' denotes homophily, 
while ``teenagers incline to interact with the ones that have opposite gender on social media'' denotes heterophily.
It is thus important to distinguish which \nef the graph has, i.e., homophily, heterophily, or both (which we call ``\xophily''), if there is any.
Given $c$ classes, an intuitive way to describe \nef is via a $c \times c$ compatibility matrix, which shows the relative influence between each class pair.
It can be used to explain the graph property, as well as be exploited to better assign the labels in the graph.

However, identifying \nef is commonly neglected in literature:
inference-based methods assume that the relationship is given by domain experts \cite{eswaran2020higher,DBLP:journals/pvldb/GatterbauerGKF15};
most graph neural networks (GNNs) assume homophily \cite{kipf2016semi,klicpera2018predict,wu2019simplifying}.
Although some previous works \cite{lim2021large,ma2022is,zhu2020beyond} use homophily statistics to analyze the given graph, our work has a very different direction because of three reasons.
First, they are designed to identify the absence of homophily, and thus can not clearly distinguish \nef, which includes different non-homophily cases, i.e., heterophily, both, or no \nef.
Second, to compute accurate statistics, they use all the node labels in the graph, which is impractical during node classification.
Finally, their analyses rely heavily on the results of GNNs, which means in addition to the graph structure, the node features also significantly influence the conclusions of \nef.
In contrast, our work aims to answer $3$ research questions:
\begin{compactenum}[{RQ}1.]
    \item {\bf Hypothesis Testing}: How to identify whether the given graph has \nef or not, with only few labels? 
    \item {\bf Estimation}: How to estimate \nef in a principled way, and explain the graph with the estimation?
    \item {\bf Exploitation}: How to efficiently exploit \nef on node classification with only few labels?
\end{compactenum}
\noindent
We propose \method, with $3$ contributions as the corresponding solutions:
\ben
\item {\bf \theory: \methodtest} 
    uses statistical tests to decide whether \nef exists at all.
    Fig.~\ref{fig:c1} shows how it works, and Fig.~\ref{fig:dis} shows its discovery, where many large real-world datasets known as heterophily graphs have little \nef.
\item {\bf \general and \explain: \methodest} explains whether the graph is homophily, heterophily, or \xophily by precisely estimating the compatibility matrix with the derived closed-form formula.
In Fig.~\ref{fig:c2}, it explains the interrelations of classes by the estimated compatibility matrix, which implies \xophily.

\item {\bf \accurate and \scale: \methodexp} efficiently exploits \nef to perform better in node classification.
It wins in both accuracy and time on a million-scale heterophily graph ``Pokec-Gender'', only requiring {\em $14$ minutes} (Fig.~\ref{fig:c3}).
\een



\noindent{\bf Reproducibility}: The code\footnote{\codeurl} and the extended version with appendix\footnote{\url{https://arxiv.org/abs/2301.00270}} are made public.


\vspace{-4mm}
\section{Background and Related Work} \label{sec:background}
\vspace{-2.5mm}



\subsection{Background}
\vspace{-2mm}
\noindent\textbf{Notation.}
Let $G$ be an undirected and unweighted graph with $n$ nodes and $m$ edges and an adjacency matrix ${\boldsymbol A}$. 
Each node $i$ has a unique label $l(i) \in \{1, 2, \dots, c\}$, where $c$ is the number of classes. 
Let ${\boldsymbol E} \in \mathbb{R}^{\n \times \classes}$ be the initial belief matrix with the prior information, i.e., the labeled nodes. 
${\boldsymbol E}_{ik} = 1$ if $l(i) = k$,
and ${\boldsymbol E}_{ik} = 0$ if $l(i) \neq k$. 
For the nodes without labels, their entries are set to $1 / c$. 
${\boldsymbol H} \in \mathbb{R}^{\classes \times \classes}$ is a row-normalized compatibility matrix, where ${\boldsymbol H}_{ku}$ is the relative influence of class $l$ on class $u$.
The residual of a matrix around $k$ is $\hat{{\boldsymbol Y}} = {\boldsymbol Y} - k\times\mathbbm{{\boldsymbol 1}}$, where 
$\mathbbm{{\boldsymbol 1}}$ is matrix of ones.

\noindent\textbf{Belief Propagation (BP).} 
\fabp~\cite{DBLP:conf/pkdd/KoutraKKCPF11} and \linbp \cite{DBLP:journals/pvldb/GatterbauerGKF15} accelerate BP by approximating the final belief assignment. 
In particular, \linbp approximates the final belief as:
\vspace{-1mm}
\begin{equation} \label{eq:prop}
\hat{{\boldsymbol B}} = \hat{{\boldsymbol E}} + {\boldsymbol A}\hat{{\boldsymbol B}}\hat{{\boldsymbol H}},
\vspace{-2mm}
\end{equation}
where $\hat{{\boldsymbol B}}$ is a residual final belief matrix, initialized with zeros.
The compatibility matrix ${\boldsymbol H}$ and initial beliefs ${\boldsymbol E}$ are centered around $1 / c$ to ensure convergence.
\hols \cite{eswaran2020higher} is a BP-based method, which propagates the labels by weighing with higher-order cliques.

\vspace{-5mm}
\subsection{Related Work}
\vspace{-2mm}
Table~\ref{tab:salesman} presents qualitative comparison of state-of-the-art approaches against our proposed \method.
Notice that only \method fullfills all the specs.

\setlength\intextsep{0pt}
\begin{wrapfigure}{L}{0.5\columnwidth}
\begin{minipage}[t]{0.5\columnwidth}
\captionof{table}{
    \emphasize{\method matches all specs}, while baselines miss one or more.
    `?' and `N/A' denote unclear and not applicable.
    \label{tab:salesman}
}
\vspace{-3mm}
\centering{\resizebox{1\columnwidth}{!}{
    \begin{tabular}{ l | l | c  c  c c | c }
        \hline
        \multicolumn{2}{c|}{\bf Property} & 
        \rotatebox{90}{BP \cite{DBLP:journals/pvldb/GatterbauerGKF15,DBLP:conf/pkdd/KoutraKKCPF11}} &
        \rotatebox{90}{HOLS \cite{eswaran2020higher}} & 
        \rotatebox{90}{General GNNs \cite{kipf2016semi,klicpera2018predict}} & 
        \rotatebox{90}{Het. GNNs \cite{abu2019mixhop,chien2021adaptive}} & 
        \rotatebox{90}{\bf \method} \\ 
        \hline
        \multirow{2}{*}{\bf 1. \theory} & 
        1.1. Statistical Test & 
        \CheckmarkBold & \CheckmarkBold & & & \CheckmarkBold \\ 
         & 
        1.2. Convergence Guarantee & 
        \CheckmarkBold & \CheckmarkBold & & & \CheckmarkBold \\ 
        \hline
        {\bf 2. \explain} & 
        2.1 Compatibility Matrix Estimation &
        N/A & N/A &  &  & \CheckmarkBold \\ 
        \hline
        \multirow{2}{*}{\bf 3. \general} &  
        3.1 Handle Heterophily &  
        ? & ? & ? & \CheckmarkBold & \CheckmarkBold \\ 
         &  
        3.2 Handle \nef &  
        ? & ? &  &  & \CheckmarkBold \\ 
        \hline
        \multirow{2}{*}{\bf 4. \scale} &
        4.1. Linear Complexity & 
        \CheckmarkBold &  & \CheckmarkBold & \CheckmarkBold & \CheckmarkBold \\
         &
        4.2. Thrifty & 
        \CheckmarkBold & \CheckmarkBold & ? & ? & \CheckmarkBold \\
        \hline
    \end{tabular}
}}
\end{minipage}
\end{wrapfigure}

\noindent\textbf{Analysis by Homophily Statistics.}
Many studies \cite{lim2021large,ma2022is,zhu2020beyond} utilize homophily ratio to measure how common the labels of the connected node pairs share the same class.
Our work focuses on very different aspects, as discussed in the introduction.

\noindent\textbf{Node Classification.}
GCN \cite{kipf2016semi} and \appnp \cite{klicpera2018predict} incorporate neighborhood information to do better predictions and assume homophily. 
\mixhop \cite{abu2019mixhop}, \gprgnn \cite{chien2021adaptive}, and $\text{H}_{2}$GCN \cite{zhu2020beyond} make no assumption of homophily.
Nevertheless, $\text{H}_{2}$GCN requires too much memory and thus can not handle large graphs.
\textsc{LINKX} \cite{lim2021large} introduces multiple large heterophily datasets, but it is not applicable to graphs without node features. 


\vspace{-5mm}
\section{Proposed \nef Test} \label{sec:netest}
\vspace{-3mm}


\begin{figure}[t]
    \captionsetup[subfloat]{farskip=2pt,captionskip=-4pt}
	\centering
	\subfloat[\label{fig:dis1} \scriptsize ``Genius'': No \nef]
	{\includegraphics[height=0.83in]{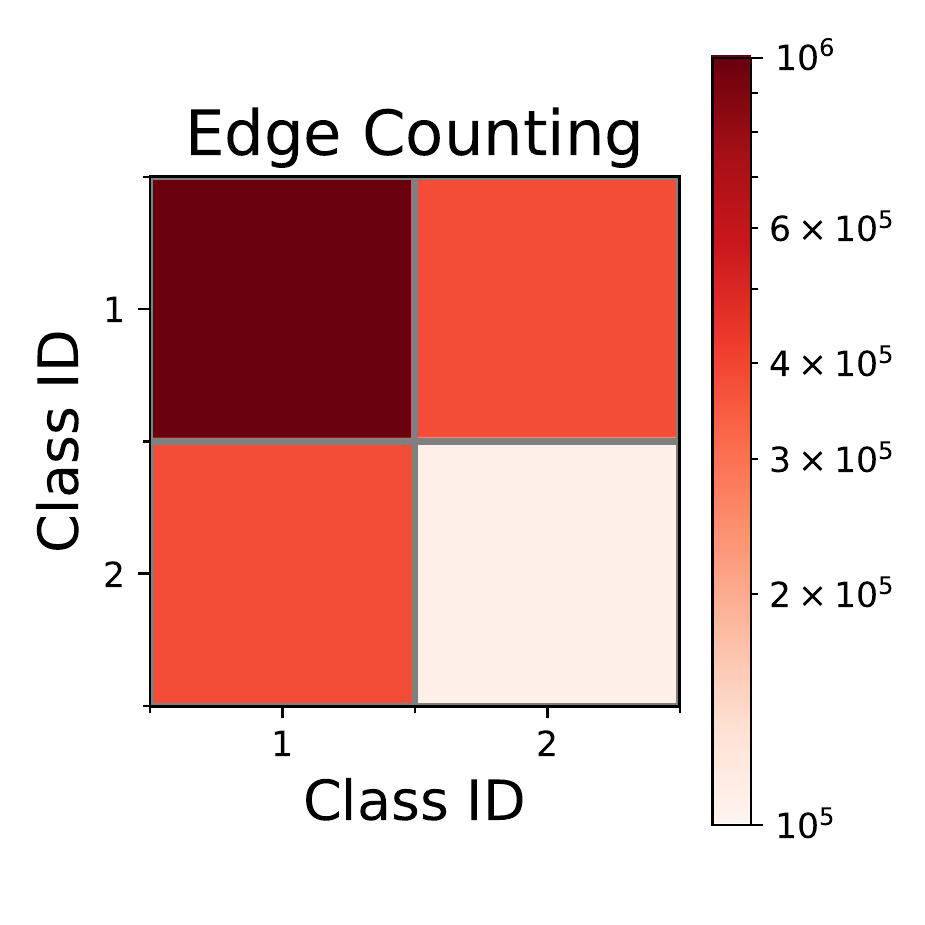}
     \hspace{-3mm}
	 \includegraphics[height=0.86in]{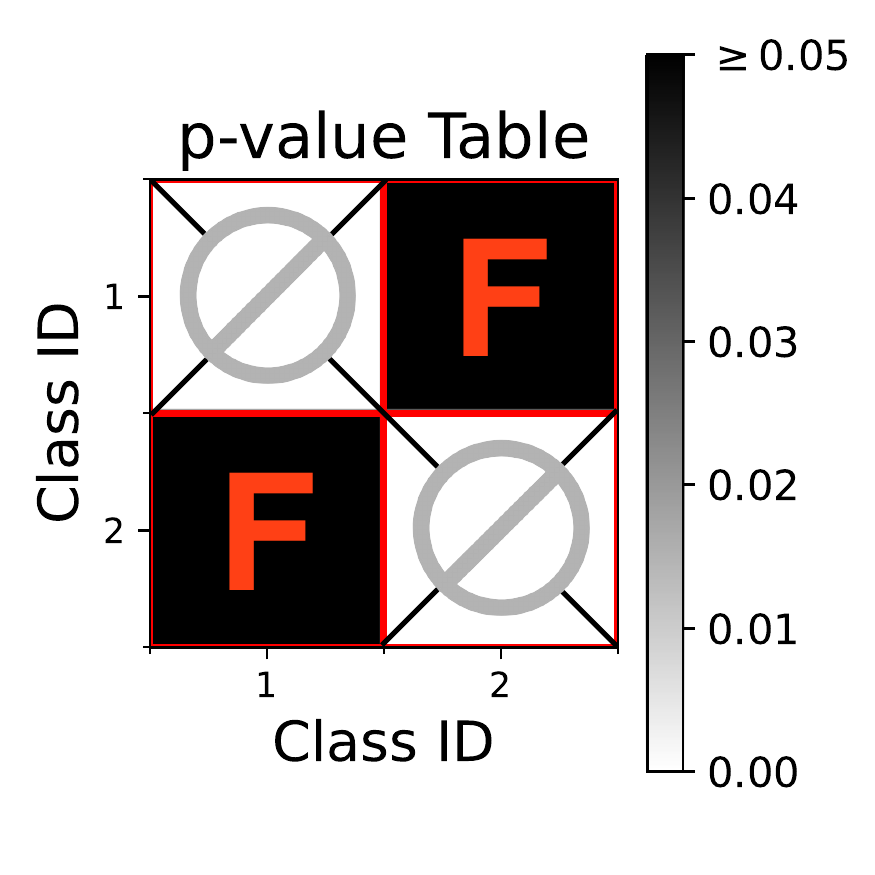}}
	 \subfloat[\label{fig:dis2} \scriptsize ``Penn94'': No \nef]
	{\includegraphics[height=0.83in]{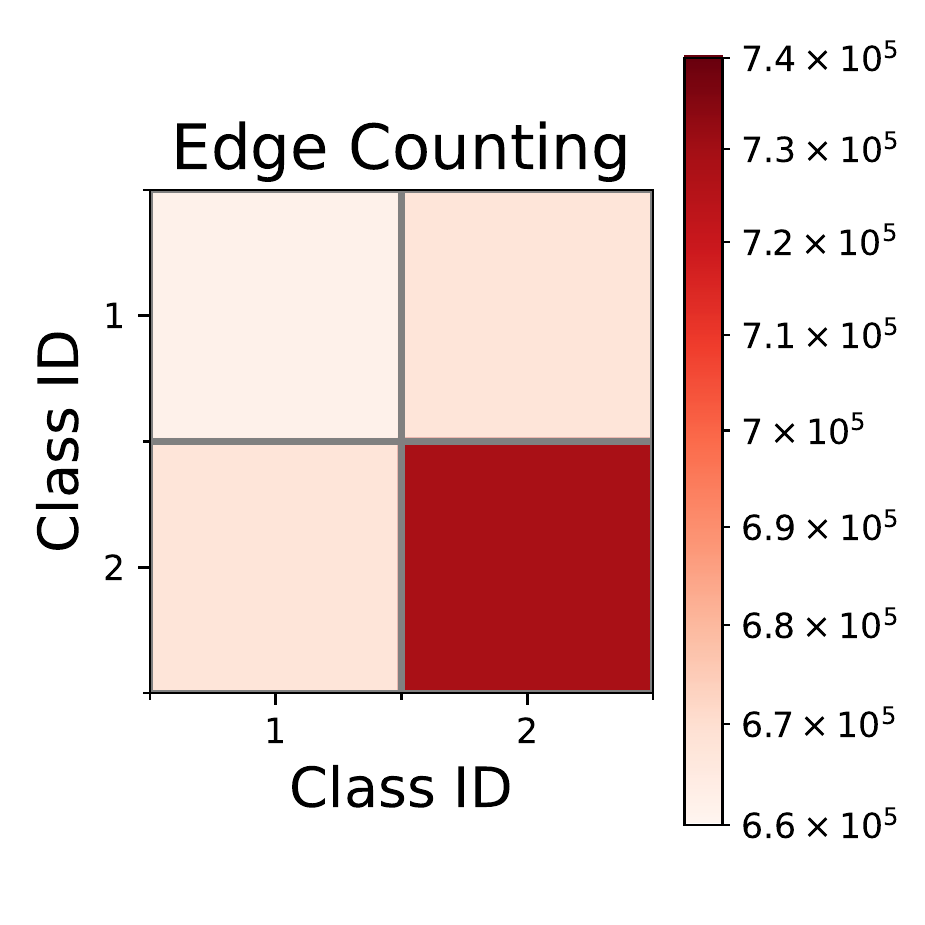}
     \hspace{-3mm}
	 \includegraphics[height=0.86in]{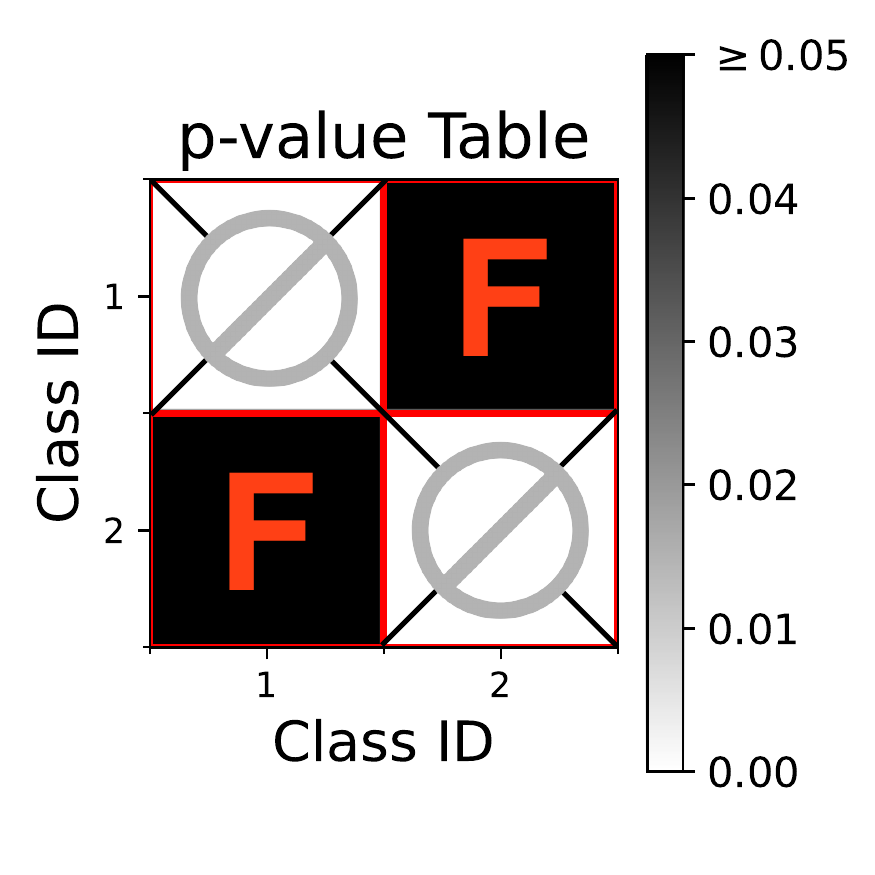}}
	 \subfloat[\label{fig:dis3} \scriptsize ``Twitch'': No \nef]
	{\includegraphics[height=0.83in]{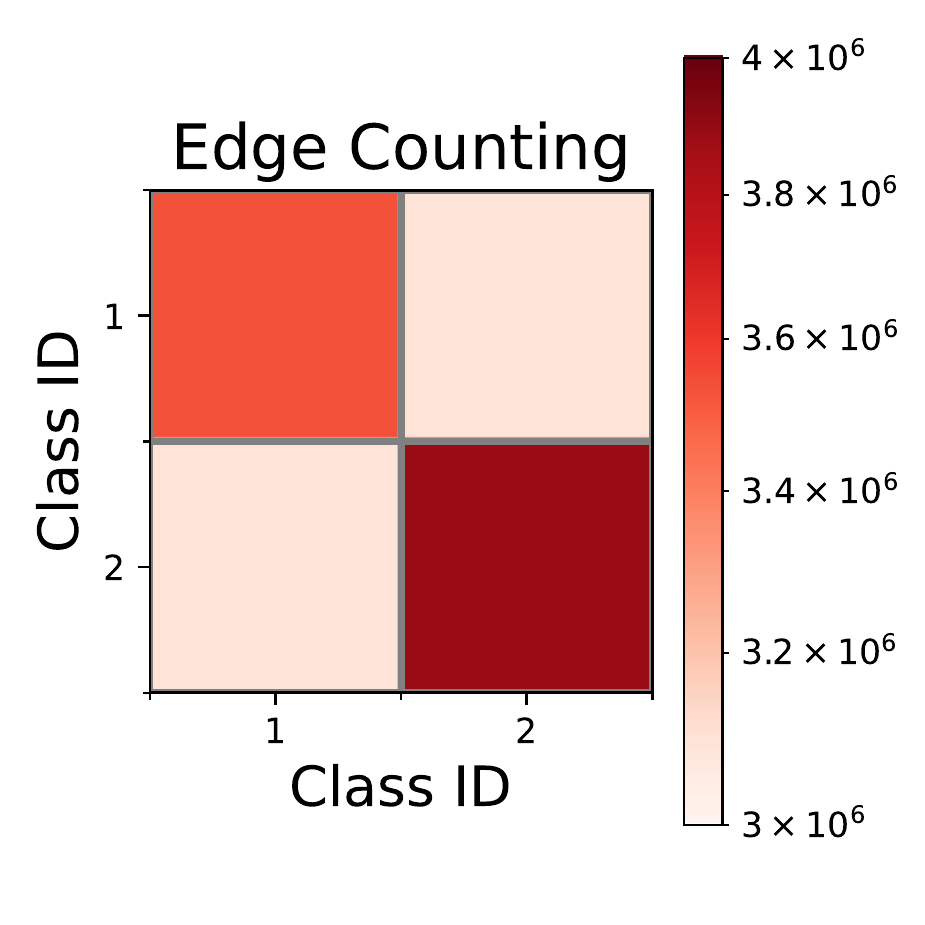}
     \hspace{-3mm}
	 \includegraphics[height=0.86in]{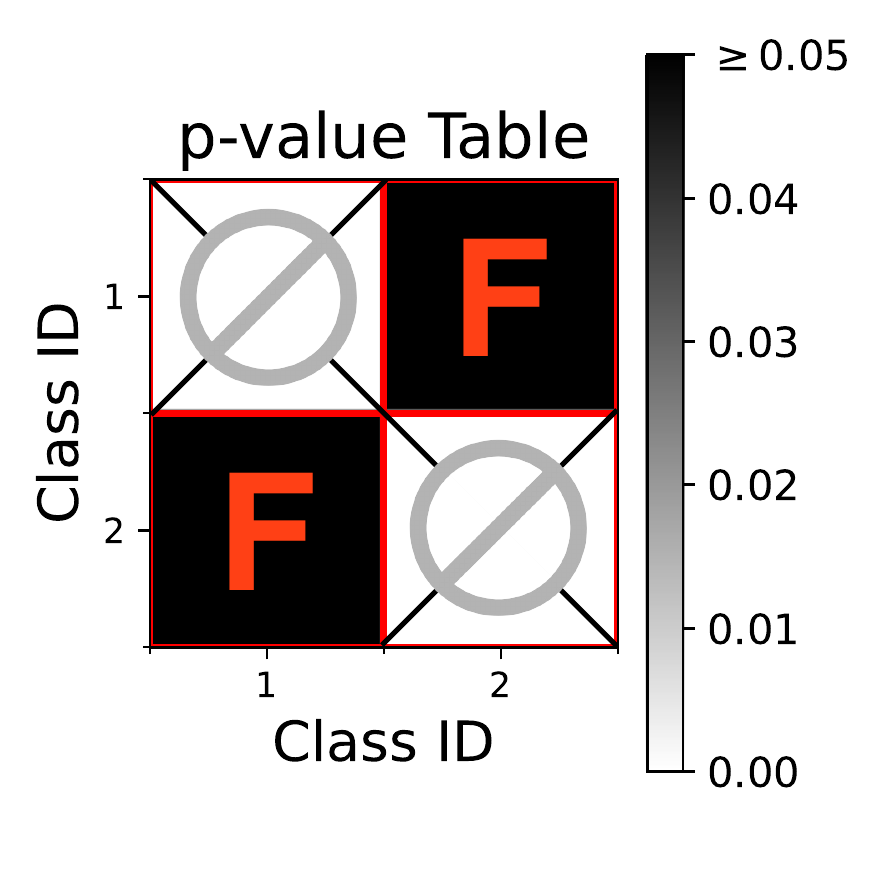}} \\
  \vspace{-2mm}
	 \subfloat[\label{fig:dis4} \scriptsize ``arXiv-Year'': Weak \nef]
	{\includegraphics[height=0.83in]{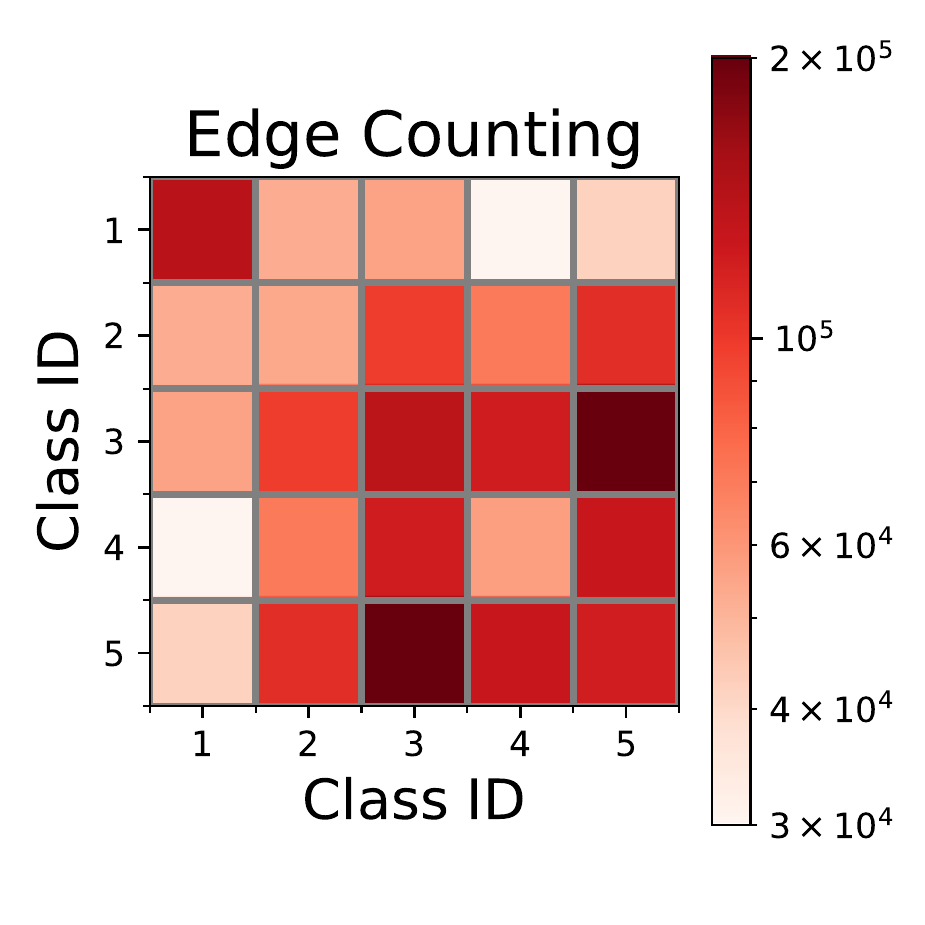}
     \hspace{-3mm}
	 \includegraphics[height=0.86in]{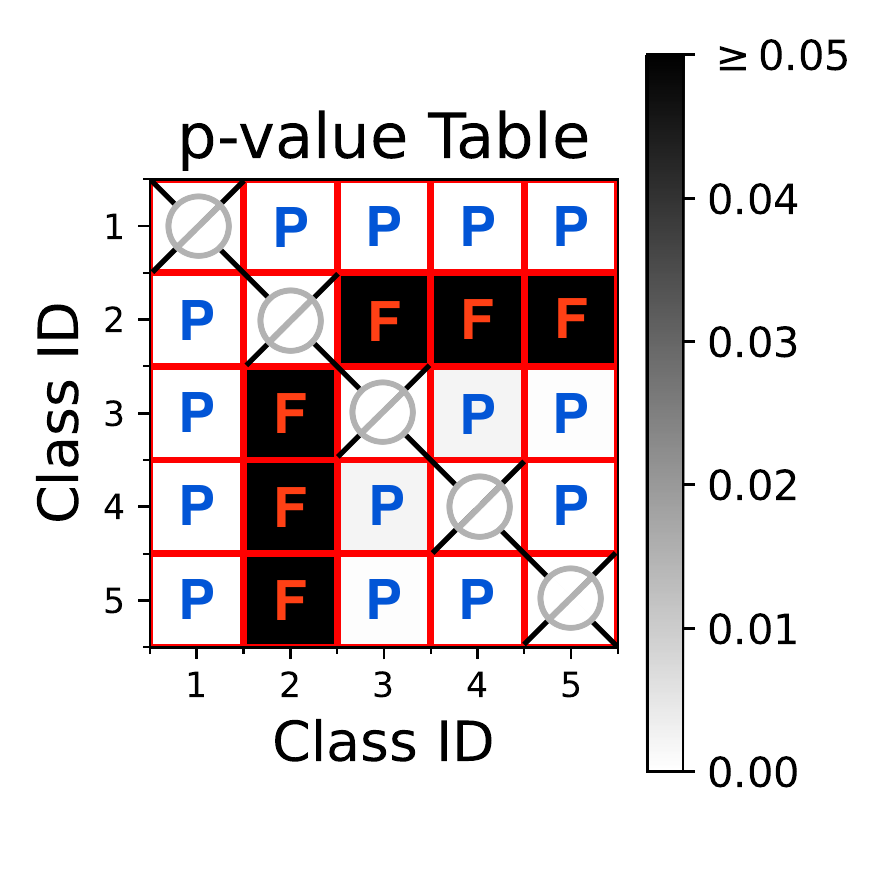}}
	 \subfloat[\label{fig:dis5} \scriptsize ``Patent-Year'': Weak \nef]
	{\includegraphics[height=0.83in]{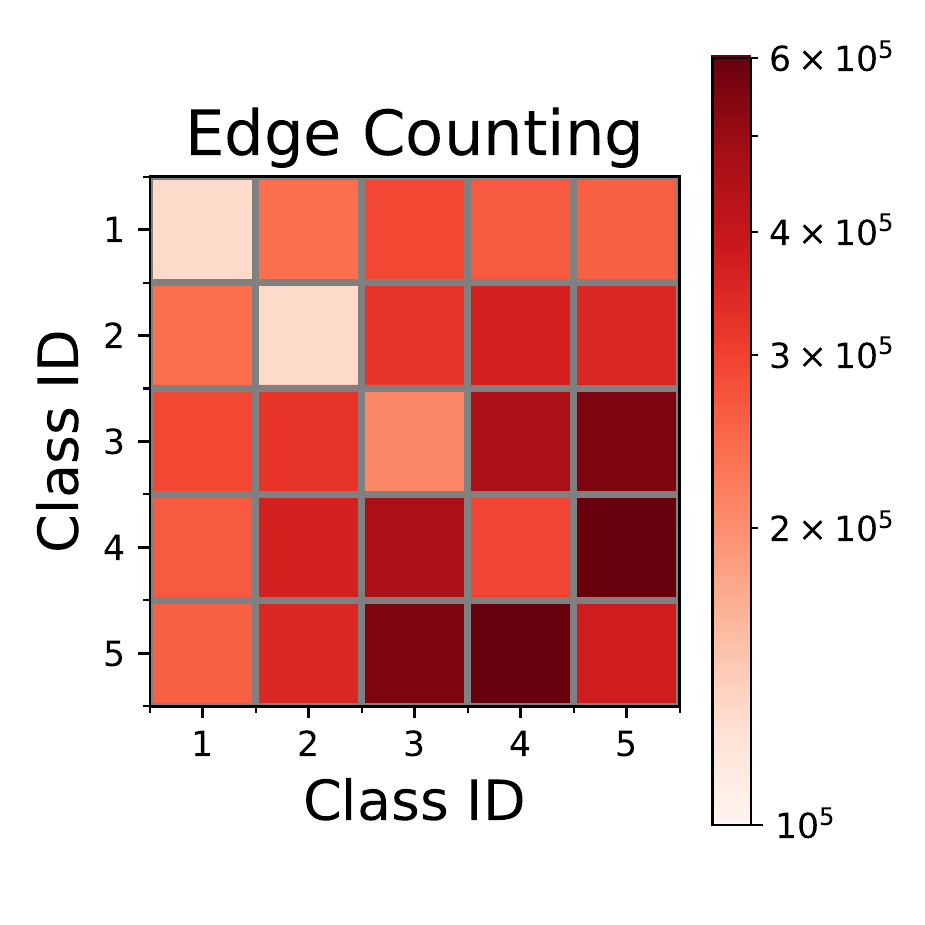}
     \hspace{-3mm}
	 \includegraphics[height=0.86in]{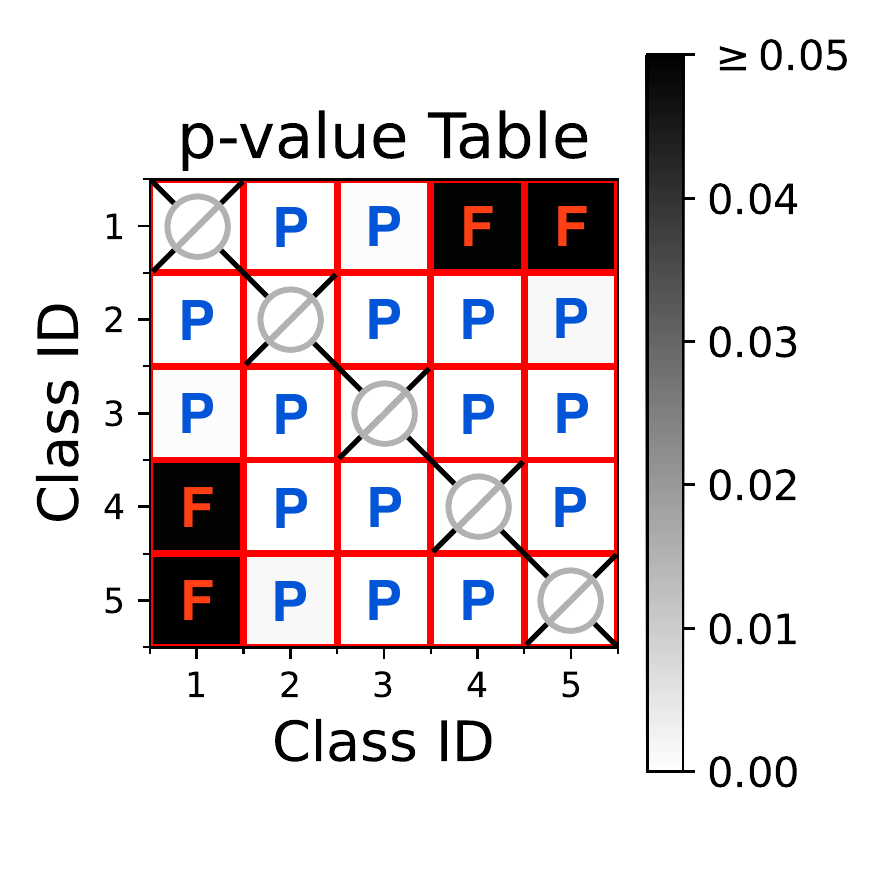}}
	 \subfloat[\label{fig:dis6} \scriptsize ``Pokec-Gender'': Strong \nef]
	{\includegraphics[height=0.83in]{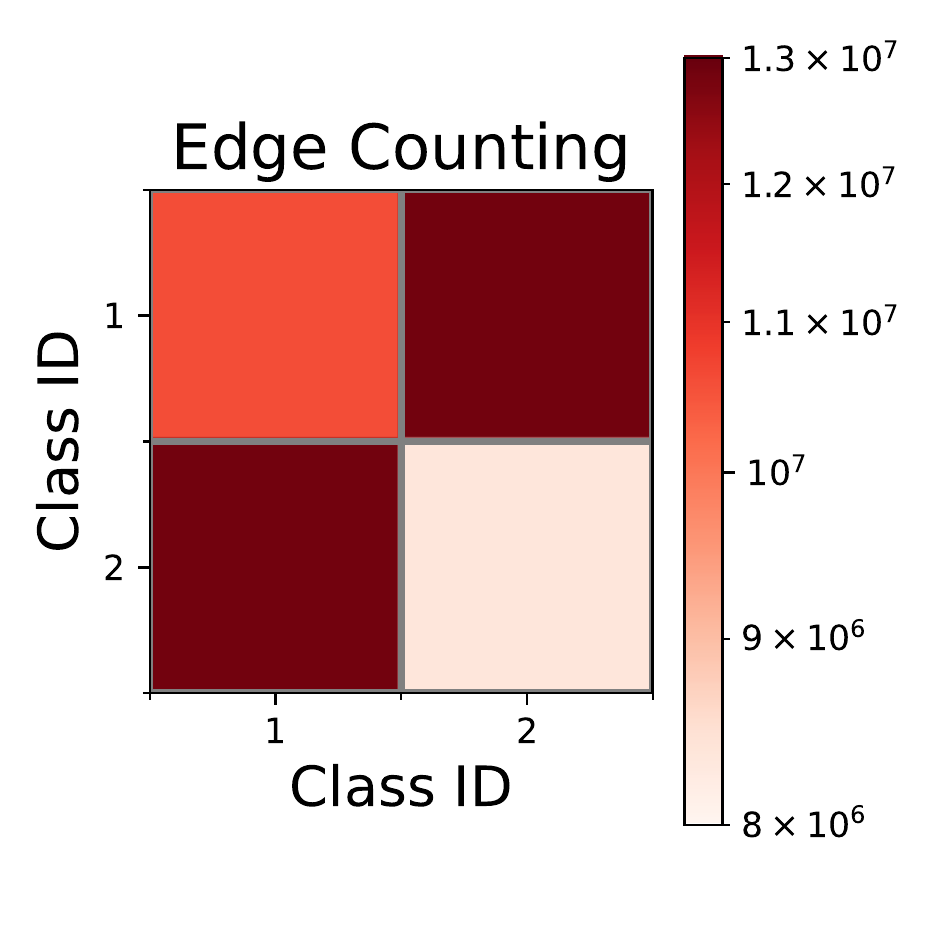}
     \hspace{-3mm}
	 \includegraphics[height=0.86in]{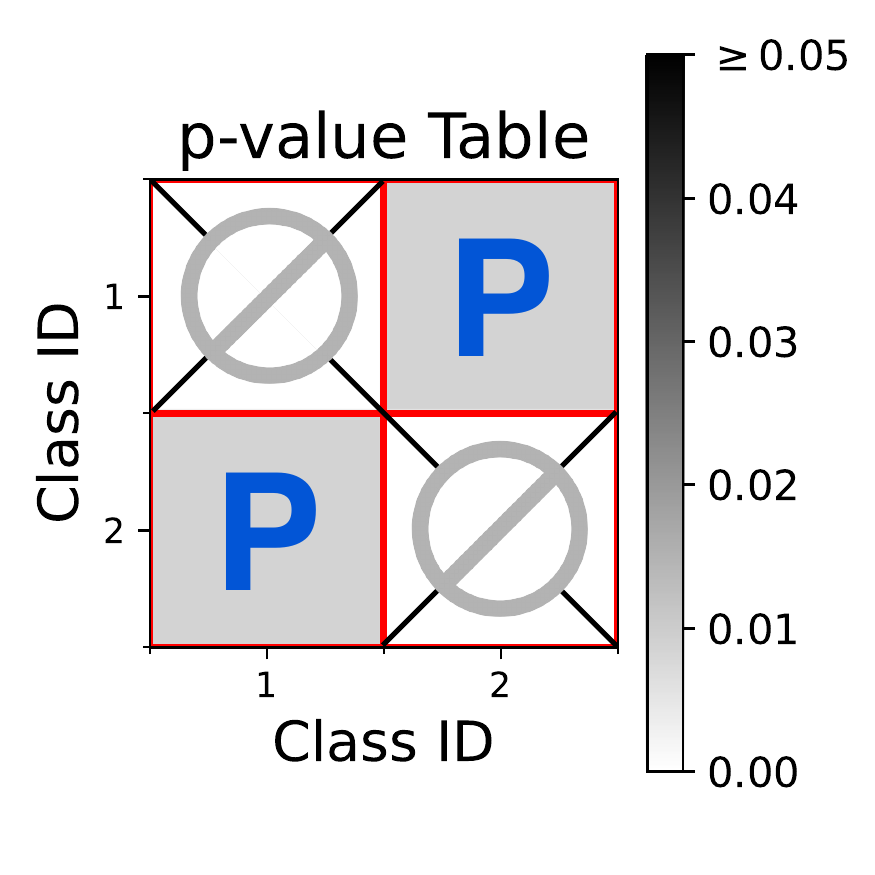}} \\
    \captionsetup[subfloat]{farskip=2pt,captionskip=0pt}
    \subfloat[\label{tab:homstat} Homophily statistics of graphs and their \nef.]
    {    
        \setlength{\tabcolsep}{1pt}
        \resizebox{0.9\columnwidth}{!}{
        \begin{tabular}{c | ccc | cc | cc}
        \hline
        \textbf{Datasets} & \textbf{Genius} & \textbf{Penn94} & \textbf{Twitch} & \textbf{Patent-Year} & \textbf{Pokec-Gender} & \textbf{arXiv-Year} & \textbf{Synthetic} \\
        \hline
        \# of Nodes / Edges / Classes & 422K / 985K / 2 & 42KM / 1.4M / 2 & 168K / 6.8M / 2 & 1.3M / 4.3M / 5 & 1.6M / 22.3M / 2 & 169K / 1.2M / 5 & 1.2M / 34.0M / 6 \\
        \hline
        Edge Homophily & 0.618 & 0.470 & 0.545 & 0.132 & 0.425 & 0.222 & 0.314 \\
        $\hat{h}$ & 0.080 & 0.046 & 0.090 & 0.000 & 0.000 & 0.272 & 0.245 \\
        \hline
        \nef & No \nef & No \nef & No \nef & Heterophily & Heterophily & \xophily & \xophily \\
        \hline
        \end{tabular}}
    }
    \vspace{-3mm}
	\caption{\label{fig:dis} \emphasize{\methodtest works}: It discovers that real-world heterophily graphs do not necessarily have \nef. For each graph, we report the edge counting on the left (not available in practice), and the $p$-value table output from \methodtest on the right, where \textcolor{blue}{``P''} denotes the presence of \nef, and \textcolor{red}{``F''} denotes the absence of \nef.
    }
    \vspace{-5mm}
\end{figure}

Given a graph with few labels, how can we identify whether the graph has \neteffect (\nef) or not?
In other words, how can we check whether the graph structure is useful for inferring node labels?
We propose \methodtest, a statistical approach to identify the presence of \nef in a graph.
Applying it to real-world graphs, we show that many popular heterophily graphs exhibit little \nef.

\subsection{\methodtest} \label{ssec:neanalysis}
\vspace{-2mm}
We first provide two main definitions regarding \nef:
\vspace{-2mm}
\begin{definition} \label{def:noGNEclass}
If the nodes with class $c_{i}$ in a graph tend to connect randomly to the nodes with all classes $1, \dots, c$ (with no specific preference), class $c_{i}$ has no \nef.
\end{definition}
\vspace{-4mm}
\begin{definition} \label{def:noGNEall}
If all classes in a graph have no \nef, this graph has no \nef.
\end{definition}

We distinguish heterophily graphs from those with no \nef by the definition.
In heterophily graphs, the nodes of a specific class are likely to be connected to the nodes of other classes, such as in bipartite graphs that connect different classes of nodes.
In this case, knowing the label of a node gives meaningful information about the labels of its neighbors.
On the other hand, if a graph has no \nef, knowing the label of a node gives no useful information about its neighbors.
In other words, the structural information of a graph is not useful to infer the unknown labels of nodes.

Next we describe how we propose to determine
the existence or absence of \nef.
In the inner loop, we need to decide whether class $c_i$ (say, ``talkative people''), has statistically more, or fewer edges
to class
$c_j$ (say, ``silent people'').
\hide{\tiny 
To identify whether a given class has \nef or not, we propose to test whether the class is distinguishable with other classes in the graph.
If it is not distinguishable, then it means the class has no \nef.}
We propose to use Pearson's $\chi^2$ test for that.
Specifically,
given a class pair $(c_i, c_j)$, the input to the test is a $2 \times 2$ contingency table containing the counts of edges that connect pairs of nodes whose labels are in $\{c_{i}, c_{j} \}$.
The null hypothesis of the test is:
\vspace{-2mm}
\begin{h0}
Edges are equally likely to exist between nodes of the same class and those of different classes.
\end{h0}
\vspace{-2mm}

\noindent
If the $p$-value from the test is no less than $0.05$, we accept the null hypothesis, which represents that the chosen class pair $(c_i, c_j)$ exhibits no statistically significant \nef in the graph.
Then we call them {\em mutually indistinguishable}:

\vspace{-2mm}
\begin{definition}[Mutually indistinguishable]
\label{def:mIndistinguishable}
Two classes $c_i$ and $c_j$ are {\em mutually indistinguishable}
if we can not reject the null hypothesis above.
\end{definition}
\vspace{-2mm}

\textbf{Novel Implementation Details.}
The detailed procedure of \methodtest is in Appx.~\ref{algo:fisher}.
A practical challenge on the test is that if the numbers in the table are too large, $p$-value becomes very small and meaningless \cite{lin2013research}.
Uniform edge sampling can be a natural solution, but sampling for only a single round can be unstable and output very different results.
To address this, we combine $p$-values from different random sampling by Universal Inference \cite{wasserman2020universal}.
We firstly sample edges to add to the contingency table until the frequency is above a specified threshold, and compute the $\chi^2$ test statistic for each class pair. 
Next, following Universal Inference, we repeat the procedure for random samples of edges for $B$ rounds and average the statistics.
At last, we use the average statistics to compute the $p$-value table ${\boldsymbol F}_{c \times c}$ of $\chi^2$ tests. 
Our \methodtest is robust to noisy edges thanks to the sampling process, and works well given either a few or many node labels.
Given a few observations, the $\chi^2$ test works well when the frequency in the contingency table is at least $5$;
given many observations, our sampling trick ensures the correctness and the consistency of the computed $p$-value.


If a class accepts the null hypotheses with all other classes, this class has little \nef, and satisfies Def.~\ref{def:noGNEclass}.
Moreover, if all classes exhibit little \nef, the whole graph satisfies Def.~\ref{def:noGNEall}.
In that case, no label propagation methods will help with node classification.



\hide{
\begin{prop} \label{obs:nf2}
If all classes in a graph obey Prop.~\ref{obs:nf}, then without extra information other than the given graph structure, exploiting \nef in node classification can not perform statistically better than random. 
\end{prop}
} 

\vspace{-5mm}
\subsection{Discoveries on Real-World Graphs} \label{ssec:discover}
\vspace{-2mm}
We apply \methodtest to $6$ real-world graphs and analyze their \nef.
For each dataset, we sample $5$\% of node labels and compute the $p$-value table using \methodtest.
This is because 
a) only few labels are available in most node classification tasks in practice, and thus it is reasonable to make the same assumption in the analysis, and 
b) \methodtest can analyze \nef even from partial observations.
$B$ is set to $1000$ to output stable results.
Based on Def.~\ref{def:noGNEall} 
our surprising discoveries are:\looseness=-1

\noindent\textbf{Discovery 1} (No \nef)
\textit{
\methodtest identifies the lack of \nef in ``Genius'' \cite{lim2020expertise}, ``Penn94'' \cite{traud2012social}, and ``Twitch'' \cite{rozemberczki2021twitch}.
}
They are widely known as heterophily graphs.
In ``Genius'' (Fig.~\ref{fig:dis1}), we see that both classes $1$ and $2$ tend to connect to class $1$, making class $2$ indistinguishable by the graph structure. 
\methodtest thus accepts the null hypothesis, and identifies the lack of \nef.
We can observe a similar phenomenon in ``Penn94'' (Fig.~\ref{fig:dis2}).
``Twitch'' (Fig.~\ref{fig:dis3}) used to be considered as a heterophily graph because of its weak homophily effect, but \methodtest finds that each of the classes uniformly connects to both classes, and thus it has little \nef.


\noindent\textbf{Discovery 2} (Heterophily and \xophily)
\textit{
\methodtest identifies \nef in ``Arxiv-Year'', ``Patent-Year'', and ``Pokec-Gender''.
}
While ``Patent-Year'' and ``Pokec-Gender'' exhibit heterophily (Fig.~\ref{fig:dis5} and \ref{fig:dis6}),
``Arxiv-Year'' exhibits \xophily, i.e., not straight homophily or heterophily (Fig.~\ref{fig:dis4}).
They are thus used in our experiments.

\noindent\textbf{Discovery 3} (Weak vs strong \nef)
\textit{
\methodtest identifies weak, and strong \nef:
``Arxiv-Year'' and ``Patent-Year'' exhibit weak \nef;
and ``Pokec-Gender'' exhibits strong \nef.
}
We consider graphs to have weak \nef if there exists at least one class which is not distinguishable from some other classes.
Such graphs limit the accuracy of node classification, compared with graphs with strong \nef (i.e., all classes have \nef), regardless of the specific method used for classification.

\noindent\textbf{Discussion of Homophily Statistics.}
In Fig.~\ref{tab:homstat}, we report two homophily statistics.
Edge homophily \cite{zhu2020beyond} is the edge ratio that connect two nodes with the same class, and $\hat{h}$ \cite{lim2021large} is an improved metric which is insensitive to the class number and size.
We find even using all labels, they are not enough to capture the interrelations of all class pairs in detail, and the graphs with low homophily statistics are not guaranteed to be heterophily.
They can only detect the absence of homophily, instead of distinguishing different non-homophily cases, including heterophily, \xophily, and no \nef.
In contrast, our \methodtest identifies whether the graph exhibits \nef or not from only a few labels.\looseness=-1

\vspace{-2mm}
\section{Proposed \nef Estimation} \label{sec:neest}
\vspace{-2mm}


Given that a graph exhibits \nef, how can we estimate the all-pair relations between classes?
A \emph{compatibility matrix} is a natural strategy to describe the relations, which has been widely used in the literature.
We propose \methodest, which turns the compatibility matrix estimation into an optimization problem based on a closed-form formula.
\methodest not only overcomes the limitation of naive edge counting, but is also robust to noisy observations even with few observed labels.

\vspace{-3mm}
\subsection{Why NOT Edge Counting}
\vspace{-1mm}
The graph in Fig.~\ref{fig:ecexgt} exhibits heterophily between class pairs $(1, 2)$ and $(3, 4)$, while it exhibits homophily in classes $5$ and $6$.
A compatibility matrix is commonly used in existing studies, but assumed given by domain experts, instead of being estimated.
A naive way to estimate it is via counting labeled edges, but it has two limitations: 
1) rare labels are neglected, and 
2) it is noisy or biased due to few labeled nodes.
The result is even more unreliable if the given labels are imbalanced.
In Fig.~\ref{fig:ecex}, we upsample the training labels $10\times$ for class $1$ using the graph in Fig.~\ref{fig:c2}.
Edge counting in Fig.~\ref{fig:ecex1} biases towards the upsampled class and clearly fails to estimate the correct compatibility matrix in Fig.~\ref{fig:ecexgt}, while our proposed \methodest succeeds in Fig.~\ref{fig:ecex2}.
This commonly occurs in practice, since we observe only limited labels, and becomes fatal if the observed distribution is different from the true one.

\vspace{-3mm}
\subsection{Closed-Form Formula} \label{ssec:comp}
\vspace{-1mm}
We begin the derivation by rewriting Eqn.~\ref{eq:prop} of BP.
The main insight is reminiscent of `leave-one-out' cross validation.
That is, we find $\hat{{\boldsymbol H}}$ that would make the results of the propagation (RHS of Eqn.~\ref{eq:simple})
to the actual values (LHS of Eqn.~\ref{eq:simple}):
\vspace{-2mm}
\begin{equation} \label{eq:simple}
\underbrace{\hat{{\boldsymbol E}}}_\text{reality} \approx
\underbrace{
{\boldsymbol A}\hat{{\boldsymbol E}}\hat{{\boldsymbol H}} }_\text{estimate}
\vspace{-2mm}
\end{equation}

\hide{
As an essential component, the compatibility matrix is used by BP for proper propagation.
Therefore, we begin the derivation by simplifying Eqn.~\ref{eq:prop} of BP.
If we initialize $\hat{{\boldsymbol B}}$ with $\hat{{\boldsymbol E}}$, and omit the addition of $\hat{{\boldsymbol E}}$ for the iterative propagation purpose, we have:
\begin{equation} \label{eq:simple_old}
\hat{{\boldsymbol B}} = {\boldsymbol A}\hat{{\boldsymbol E}}\hat{{\boldsymbol H}}
\end{equation}
}

\hide{
Our goal is to estimate the compatibility matrix $\hat{{\boldsymbol H}}$ of a given graph, so that the difference between belief propagated by the given priors $\mathcal{P}$ and the final belief is minimized.
Nevertheless, the final belief $\hat{{\boldsymbol B}}$ is not available before we run the propagation on the graph.
To address that, we mimic the leave-one-out cross-validation, by speculating the label of the node through its labeled neighbors, which turns Eqn.~\ref{eq:simple} into $\hat{{\boldsymbol E}} \approx {\boldsymbol A}\hat{{\boldsymbol E}}\hat{{\boldsymbol H}}$.
In other words, we aim to minimize the difference between initial belief of each node $i \in \mathcal{P}$ by the initial beliefs of its neighbors $N(i) \in \mathcal{P}$, i.e., $N(i) \cap \mathcal{P}$. 
Intuitively, the neighbors are able to estimate the belief for the node.
}

\noindent 
Formally, we want to minimize the difference between the reality and the estimate:
\vspace{-3mm}
\begin{equation} \label{eq:opt}
    \min_{\hat{\boldsymbol H}}\sum_{i \in \mathcal{P}}{\sum_{u=1}^{c} \|{\hat{\boldsymbol E}_{iu}} - \sum_{k=1}^{c}{\sum_{j \in N(i) \cap \mathcal{P}}{\hat{\boldsymbol E}_{jk}}\hat{\boldsymbol H}_{ku}}}\|^2,
\vspace{-2mm}
\end{equation}
where $N(i)$ denotes the neighbors of node $i$. 
In other words, we aim to minimize the difference between initial belief $\hat{\boldsymbol E}$ of each node $i \in \mathcal{P}$ by the ones of its neighbors $N(i) \in \mathcal{P}$, i.e., $N(i) \cap \mathcal{P}$. 
To estimate the compatibility matrix $\hat{{\boldsymbol H}}$, we solve the optimization problem in Eqn.~\ref{eq:opt} with the proposed closed-form formula:
\begin{lemma}[\NEF (NEF)] \label{lem:nef}
Given adjacency matrix ${\boldsymbol A}$ and initial beliefs $\hat{{\boldsymbol E}}$, the closed-form solution of vectorized compatibility matrix $\text{vec}{(\hat{{\boldsymbol H}})}$ is:
\vspace{-2mm}
\begin{equation}
\boxed{
    \text{vec}{(\hat{{\boldsymbol H}})} = ({\boldsymbol X}^{T}{\boldsymbol X})^{-1}{\boldsymbol X}^{T}{\boldsymbol y}
    }
\vspace{-2mm}
\end{equation}
where ${\boldsymbol X} = {\boldsymbol I}_{c \times c} \otimes ({\boldsymbol A}\hat{{\boldsymbol E}})$ and ${\boldsymbol y} = \text{vec}{(\hat{{\boldsymbol E}})}$.
\end{lemma}
\vspace{-2mm}
\begin{proof}
See Appx.~\ref{ap:subsec:proofprop1}.
\hfill$\blacksquare$
\end{proof}


\begin{figure}[t]
\begin{minipage}[] {0.55\linewidth}
    \centering
    \subfloat[\label{fig:ecexgt} \scriptsize Ground Truth]
    {\includegraphics[height=0.9in]{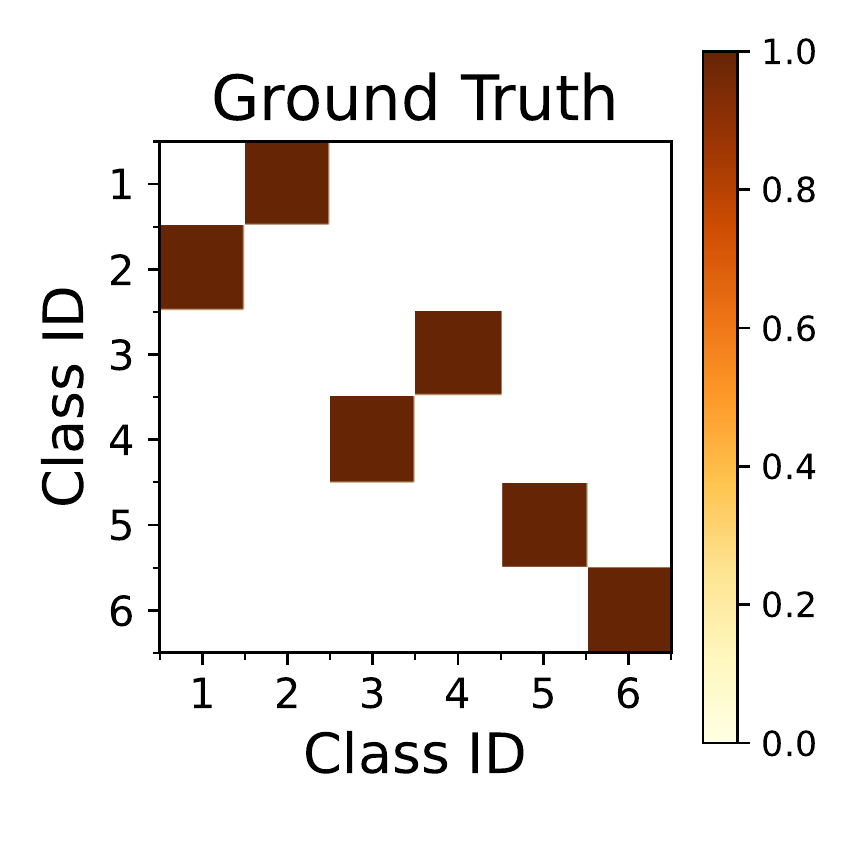}}
    \subfloat[\label{fig:ecex1} \scriptsize Edge Counting]
    {\includegraphics[height=0.9in]{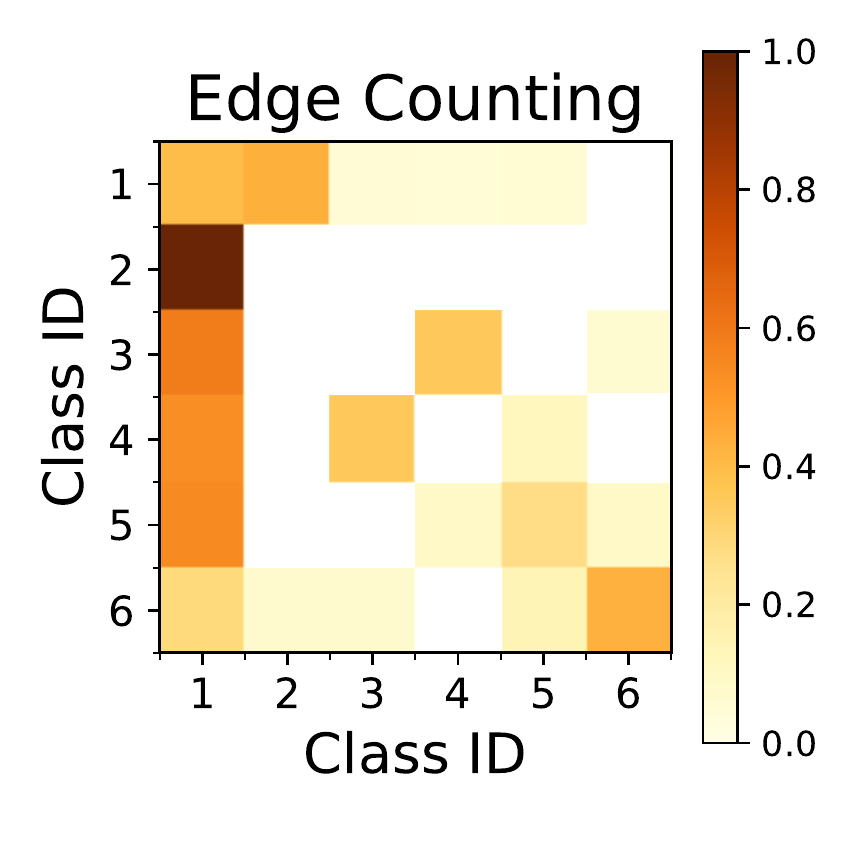}}
    \subfloat[\label{fig:ecex2} \scriptsize \methodest]
    {\includegraphics[height=0.9in]{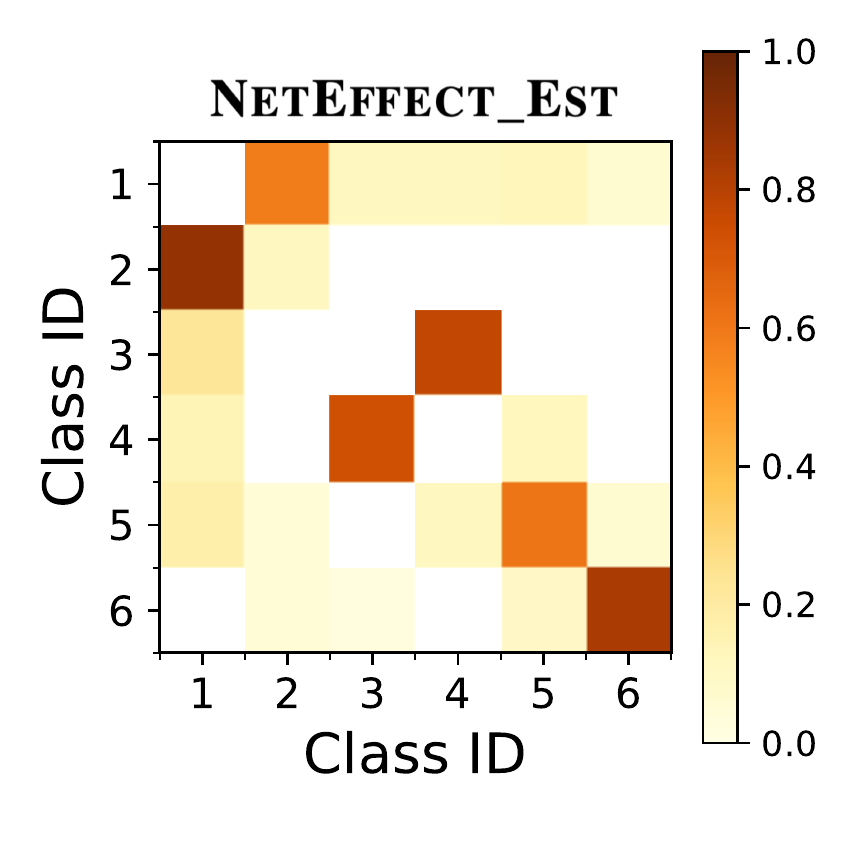}}
    \vspace{-3mm}
    \caption{\label{fig:ecex} \emphasize{\methodest handles imbalanced} \emphasize{case well}. Labels of class $1$ is upsampled.}
\end{minipage} 
\hfill
\begin{minipage}[] {0.42\linewidth}
    \centering
    \subfloat[\label{fig:em1} \scriptsize Adj. Matrix]
    {\includegraphics[height=0.86in]{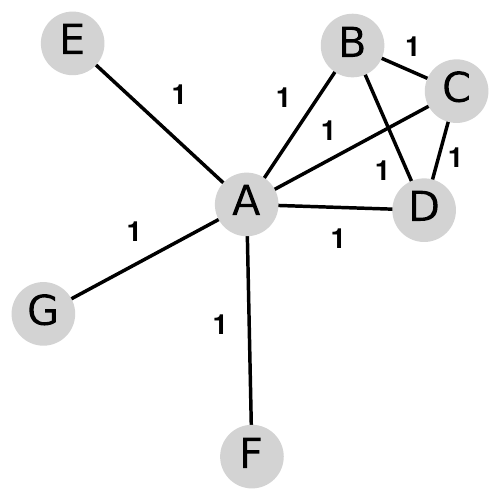}}
    \subfloat[\label{fig:em2} \scriptsize Emphasis Matrix]
    {\includegraphics[height=0.86in]{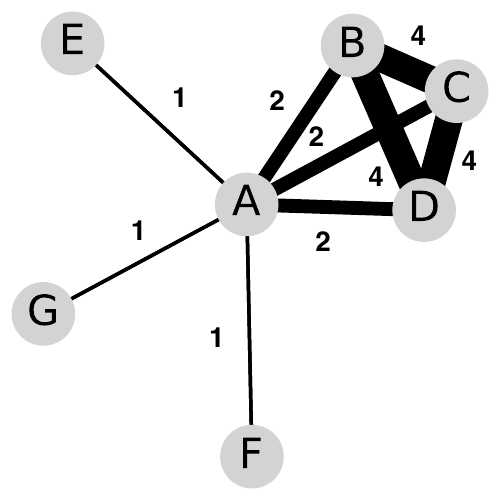}} 
    \vspace{-3mm}
    \caption{\label{fig:ematrix} \emphasize{Emphasis matrix at work}: it prefers well-connected neighbors.}
\end{minipage}
\vspace{-6mm}
\end{figure}

\subsection{\methodest}
\vspace{-3mm}
The algorithm is presented in Alg.~\ref{algo:cmest}.
In practice, we can use any form of adjacency matrix for the estimation.
The proposed NEF allows us to estimate the compatibility matrix by solving this optimization problem, but there still exists a practical challenge that need to be addressed.
With few labels, it is difficult to properly separate them into training and validation sets for the regression, and the estimation can easily be interfered by the noisy observations.
We thus use ridge regression with leave-one-out cross-validation (RidgeCV) instead of the regular linear regression.
This allows us to fully utilize the observations without having biases caused by random splits of training and validation sets.
Moreover, the regularization effect of RidgeCV makes the compatibility matrix more robust to noisy observations.
It is noteworthy that its computational cost is negligible.\looseness=-1




\scalebox{0.75}{
\begin{minipage}{0.8\linewidth}
\begin{algorithm}[H]
\KwData{Adjacency matrix ${\boldsymbol A}$, initial belief $\hat{{\boldsymbol E}}$, and priors $\mathcal{P}$}
\KwResult{Estimated compatibility matrix $\hat{{\boldsymbol H}}$}
${\boldsymbol X} \leftarrow {\boldsymbol I}_{c \times c} \otimes ({\boldsymbol A}\hat{{\boldsymbol E}})$\textcolor{blue}{\tcp*{feature matrix}}
${\boldsymbol y} \leftarrow \text{vec}{(\hat{{\boldsymbol E}})}$\textcolor{blue}{\tcp*{target vector}}
Extract indices ${\boldsymbol i}$ with nodes in priors $\mathcal{P}$\;
$\hat{{\boldsymbol H}} \leftarrow RidgeCV({\boldsymbol X}[{\boldsymbol i}], {\boldsymbol y}[{\boldsymbol i}])$\;
Return $\hat{\boldsymbol H}$\;
\caption{\methodest \label{algo:cmest}}
\end{algorithm}
\end{minipage}
}

\vspace{-4mm}
\section{Proposed \nef Exploitation} \label{sec:neexp}
\vspace{-3mm}





We propose \methodexp to exploit \nef for accurate and fast node classification with few labels.
With few labels, it becomes crucial to better utilizing the graph structure.
First, we address this by paying attention to influential neighbors by the proposed \emphasis;
and then describe \methodexp with theoretical analysis.\looseness=-1

\vspace{-5mm}
\subsection{``Emphasis'' Matrix} \label{ssec:att}
\vspace{-2mm}
\textbf{Rationale and Overview.}
With few priors, we propose to better utilize the graph structure, by paying attention to only the most important part of it.
That is to say, not all neighbors are equally influential:
In Fig.~\ref{fig:ematrix}, best practice shows that well-connected neighbors (i.e., nodes `B', `C', and `D') have more influence on node `A' than the rest.
Thus, 
we propose \emphasis ${\boldsymbol A}^{*}$, 
to pay attention to such neighbors.
\methodest can also benefit from it by replacing ${\boldsymbol A}$ with ${\boldsymbol A}^{*}$, where we denote the improved compatibility matrix as $\hat{\boldsymbol{H}}^{*}$.
Alg.~\ref{algo:rw} shows the details.
In short, it has $3$ steps:\looseness=-1
\ben
\item {\bf Favors influential neighbors} by quickly approximating the node-to-node proximity using (non-backtracking) random walks with restarts (lines $2$-$5$);
\item {\bf Touches-up} the new node-proximity by applying a series of transformations (including the best-practice element-wise logarithm) on the proximity matrix (line $6$);
\item {\bf Symmetrizes and weighs} the adjacency matrix with structural-aware embedding (lines $7$-$8$), giving higher weights to neighbors with closer embeddings (line $9$).
\een

\scalebox{0.75}{
\begin{minipage}{1.06\linewidth}
\begin{algorithm}[H]
\KwData{Adjacency matrix ${\boldsymbol A}$, number of trials $M$, number of steps $L$, and dimension $d$}
\KwResult{Emphasis matrix ${\boldsymbol A}^{*}$}
${\boldsymbol W}^{\prime} \leftarrow {\boldsymbol O}_{n \times n}$\;
\tcc{approximate proximity matrix by random walk}
\For{node $i$ in $G$}{
    \For{$m = 1, ..., M$}{
        \For{$j \in \mathcal{W}_{m}(i, L)$}{
            ${\boldsymbol W}^{\prime}_{ij} \leftarrow {\boldsymbol W}^{\prime}_{ij} + 1$\;
        }
    }
}
\tcc{masking, degree normalization and logarithm}
${\boldsymbol W}_{n \times n} \leftarrow \log{({\boldsymbol D}^{-1}({\boldsymbol W}^{\prime} \odot {\boldsymbol A}))}$\; 
${\boldsymbol U}_{n \times d}, {\boldsymbol \Sigma}_{d \times d}, {\boldsymbol V}_{d \times n}^{T} \leftarrow \text{SVD}({\boldsymbol W}, d)$\tcp*{embedding}
${\boldsymbol U} \leftarrow \sqrt{{\boldsymbol \Sigma}}{\boldsymbol U}$\tcp*{scaling}
\tcc{boost weights of close-embedded neighbors}
Weigh ${\boldsymbol A}^{*}_{n \times n}$, where ${\boldsymbol A}^{*}_{ij} = \mathcal{S}({\boldsymbol U}_{i}, {\boldsymbol U}_{j}), \forall \{i, j | {\boldsymbol A}_{ij} = 1\}$\;
Return ${\boldsymbol A}^{*}$\;
\caption{``Emphasis'' Matrix \label{algo:rw}}
\end{algorithm}
\end{minipage}
}

\hide{
which we describe in a flashback way to give a better train of thought.
The main idea is that the neighbors of a node are more influential if they are more structurally similar to that node.
In line $12$, we construct the \emphasis by giving larger weights to edges between similar nodes in the embedding space.
The node representations can be obtained by decomposing a higher-order proximity matrix, but it is usually dense and may be biased towards nodes with large degrees.
To address these issues, we derive node representations by (a) approximating the higher-order proximity matrix by random walk (in line $1$-$8$) and (b) applying a series of transformations on the resulting proximity matrix.
}

\vspace{2mm}
\noindent\textbf{Proximity Matrix Approximation.}
We propose to utilize random walks to approximate the proximity matrix.
The approximated proximity matrix ${\boldsymbol W}^{\prime}_{ij}$ records the times we visit node $j$ if we start a random walk from node $i$.
Only the well-connected neighbors will be visited more often.
\hide{
In practice, it is common that a node may have many neighbors, but most of them provide useless information for inferring its label.
To address this and approximate ${\boldsymbol W}$ in a fast way, we utilize random walks.
Given an approximated proximity matrix ${\boldsymbol W}^{\prime}$, ${\boldsymbol W}^{\prime}_{ij}$ records the number of times we visit node $j$ if we start a random walk from node $i$. 
Each neighbor has the same probability of being visited, but only those structurally important ones are visited more frequently.
} 
We theoretically show that it converges quickly:\looseness=-1
\vspace{-1mm}
\begin{lemma} [Convergence of Random Walks] \label{lem:crw1}
With probability $1 - \delta$, the error $\epsilon$ between the approximated and true distributions for a node walking to its 1-hop neighbor by random walks of length $L$ with $M$ trials is no greater than
$\frac{\lceil (L - 1) / 2 \rceil}{L} \sqrt{\frac{\log{(2/\delta)}}{2LM}}$.
\end{lemma}


We can make the convergence even faster by using ``non-backtracking'' random walks \cite{alon2007non}.
Given the start node $s$ and walk length $L$, its function is defined as follows:
\vspace{-3mm}
\begin{equation}
\mathcal{W}(s, L) = 
\begin{cases} (w_{0}=s, ..., w_{L}) \space 
\begin{array}{c}
w_{l} \in N{(w_{l-1})}, \forall l \in [1, L]
\\
w_{l-1} \neq w_{l+1}, \forall l\in [1, L-1]
\end{array}.
\end{cases}
\vspace{-2mm}
\end{equation}
Thanks to it, we improve Lemma~\ref{lem:crw1} to have a tighter bound of error $\epsilon$:
\begin{lemma} [Convergence of Non-Backtracking Random Walks] \label{lem:crw2}
With the same condition as in Lemma~\ref{lem:crw1}, the error $\epsilon$ by non-backtracking random walks is no greater than
$\frac{\lceil (L - 1) / 3 \rceil}{L} \sqrt{\frac{\log{(2/\delta)}}{2LM}}$.
\end{lemma}

\begin{proof}
See Appx.~\ref{ap:subsec:prooflemma12}. \hfill$\blacksquare$
\end{proof}


\noindent\textbf{Structural-Aware Node Representation.}
Based on ${\boldsymbol W}$, we apply a series of transformations to generate better and unbiased representations of nodes in a fast way.
An element-wise multiplication by ${\boldsymbol A}$ is done to keep the approximation of $1$-hop neighbor for each node, which is sparse but supplies sufficient information.
We use the inverse of the degree matrix ${\boldsymbol D}^{-1}$ to reduce the influence of nodes with large degrees.
This prevents them from dominating the pairwise distance by containing more elements in their rows. 
The element-wise logarithm rescales the distribution in ${\boldsymbol W}$, in order to enlarge the difference between smaller structures.
We use Singular Value Decomposition (SVD) for efficient rank-$d$ decomposition of sparse ${\boldsymbol W}$, and multiply the left-singular vectors ${\boldsymbol U}$ by the squared eigenvalues $\sqrt{{\boldsymbol \Sigma}}$ to correct the scale.

\noindent\textbf{``Emphasis'' Matrix Construction.}
Directly measuring the node similarity in the graph is not trivial, or may be time consuming (e.g., by counting motifs).
Therefore, we propose to compute the node similarity via the structural-aware node representations, which capture the higher-order information, and construct the \emphasis ${\boldsymbol A}^{*}$ by weighing ${\boldsymbol A}$ with the node similarity.
The intuition is that the nodes that are closer in the embedding space are better connected with higher-order structures. 
The similarity function is
$\mathcal{S}({\boldsymbol U}_{i}, {\boldsymbol U}_{j}) = e^{-\mathcal{D}({\boldsymbol U}_{ik}, {\boldsymbol U}_{jk})}$, where $e$ is the Euler's number. 
It is a universal law \cite{shepard1987toward}, which turns the distance into similarity, and bounds it from $0$ to $1$. 
While $\mathcal{D}$ can be any distance metric, we use Euclidean as it works well empirically.

\vspace{-4mm}
\subsection{\methodexp}
\vspace{-1mm}
The algorithm of \methodexp is in Appx.~\ref{algo:main}.
\methodexp takes as input the \emphasis ${\boldsymbol A}^{*}$,
the compatibility matrix $\hat{{\boldsymbol H}}^{*}$ estimated by ${\boldsymbol A}^{*}$, and the initial beliefs $\hat{{\boldsymbol E}}$.
It computes
the beliefs $\hat{{\boldsymbol B}}$ iteratively 
by aggregating the beliefs of neighbors through ${\boldsymbol A}^{*}$ until they converge.
This reusage of ${\boldsymbol A}^{*}$ aims to draw attention to the neighbors that are more structurally important. 
By exploiting \nef with $\hat{{\boldsymbol H}}^{*}$, \method propagates properly in heterophily graphs.


\noindent\textbf{Convergence Guarantee.}
To ensure the convergence of \methodexp, we introduce a scaling factor $f$ during the iterations. 
A smaller $f$ leads to a faster convergence but distorts the results, thus we set $f$ to $0.9 / \rho{({\boldsymbol A}^{*})}$. 
Its exact convergence is:
\begin{lemma}[Exact Convergence] \label{lem:con}
The criterion for the exact convergence of \methodexp is $0 < f < 1/\rho{({\boldsymbol A}^{*})}$, where $\rho{(\cdot)}$ denotes the spectral radius of the matrix.
\end{lemma}
\vspace{-1mm}
\begin{proof}
See Appx.~\ref{ap:subsec:proofcon}. \hfill$\blacksquare$
\end{proof}


\noindent\textbf{Complexity Analysis.}
\methodexp uses sparse matrix representation of graphs and scales linearly.
Its complexity is:
\begin{lemma} \label{lem:complexity}
The time complexity of \methodexp is approximately $O(m)$ and the space complexity is $O(\max{(m, n \cdot L \cdot M)} + n \cdot c^{2})$.
\vspace{-1mm}
\end{lemma}
\begin{proof}
See Appx.~\ref{ap:subsec:proofcomplexity}.
\hfill$\blacksquare$
\end{proof}

\vspace{-5mm}
\section{Experiments} \label{sec:exp}
\vspace{-2mm}




In this section, we aims to answer the following questions:
\begin{compactenum}[{Q}1.]
\item {\bf Accuracy}: How well does \method work by estimating and exploiting \nef?
\item {\bf Scalability}: How does the running time of \method scale w.r.t. graph size?
\item {\bf Explainability}: How does \method explain the real-world graphs?
\end{compactenum}

\noindent\textbf{Datasets.}
We focus on large graphs and include $8$ graphs with at least $20$K nodes.
For each dataset, we sample only a few node labels for training for five times and report the average.
``Synthetic'' is the enlarged graph in Fig.~\ref{fig:c2}, which exhibits \xophily \nef.\looseness=-1

\noindent\textbf{Baselines.}
We compare \method with five baselines and separate them into four groups: \textit{General GNNs:} \gcn~\cite{kipf2016semi}, \appnp~\cite{klicpera2018predict}. \textit{Heterophily GNNs:} \mixhop~\cite{abu2019mixhop}, \gprgnn~\cite{chien2021adaptive}. \textit{BP-based methods:} \hols~\cite{eswaran2020higher}. \textit{Our proposed methods:} \methodhom and \method.
\methodhom is \method using identity matrix as compatibility matrix, which assumes homophily and does not handle \nef.

\noindent\textbf{Experimental Settings.}
For GNNs, one-hot node degrees are used as the node features, as implemented by PyG \cite{Fey/Lenssen/2019}.
Experiments are run on a server with $3.2$GHz Intel Xeon CPU.
Details of the experimental setup are in Appx.~\ref{sec:rep}.

\begin{table}[!t]
\caption{\emphasize{\method wins on \xophily and Heterophily datasets.} \label{tab:effecthet}}
\centering{\resizebox{0.97\textwidth}{!}{
\begin{tabular}{c | c r r | c r r | c r r | c r r}

\hline

{\bf Dataset} 
& \multicolumn{3}{c|}{\bf Synthetic} & \multicolumn{3}{c|}{\textbf{Pokec-Gender} \cite{takac2012data}} 
& \multicolumn{3}{c|}{\textbf{arXiv-Year} \cite{hu2020open}} & \multicolumn{3}{c}{\textbf{Patent-Year} \cite{leskovec2005graphs}} \\

\hline

\# of Nodes / Edges / Classes
& \multicolumn{3}{c|}{1.2M / 34.0M / 6} & \multicolumn{3}{c|}{1.6M / 22.3M / 2}
& \multicolumn{3}{c|}{169K / 1.2M / 5} & \multicolumn{3}{c}{1.3M / 4.3M / 5} \\

Label Fraction
& \multicolumn{3}{c|}{4\%} & \multicolumn{3}{c|}{0.4\%}
& \multicolumn{3}{c|}{4\%} & \multicolumn{3}{c}{4\%}
\\

\hline

\nef Strength
& \multicolumn{3}{c|}{Strong \xophily} & \multicolumn{3}{c|}{Strong Heterophily}
& \multicolumn{3}{c|}{Weak \xophily} & \multicolumn{3}{c}{Weak Heterophily} \\



\hline

Method
& Accuracy & Time (s) & Rel. Time
& Accuracy & Time (s) & Rel. Time
& Accuracy & Time (s) & Rel. Time
& Accuracy & Time (s) & Rel. Time
\\

\hline

\gcn   
& 16.7$\pm$0.0 & 3456 & \bronze{4.1$\times$}
& 51.8$\pm$0.1 & \bronze{2906} & \silver{3.4$\times$}
& 35.3$\pm$0.1 & \silver{132} & \silver{2.5$\times$}
& 26.0$\pm$0.0 & 894 & 2.3$\times$
\\

\appnp   
& 18.6$\pm$1.1 & 7705 & 9.2$\times$
& 50.9$\pm$0.3 & 6770 & 7.8$\times$
& 33.5$\pm$0.2 & 423 & 8.1$\times$
& \silver{27.5$\pm$0.2} & 2050 & 5.2$\times$
\\

        
\hline

\mixhop  
& 16.7$\pm$0.0 & 58391 & 70.0$\times$
& 53.4$\pm$1.2 & 53871 & 62.1$\times$
& \gold{39.6$\pm$0.1} & 2983 & 57.4$\times$
& \bronze{26.8$\pm$0.1} & 18787 & 47.6$\times$
\\

\gprgnn
& 18.9$\pm$1.2 & 7637 & 9.1$\times$
& 50.7$\pm$0.2 & 6699 & \bronze{7.7$\times$}
& 30.1$\pm$1.4 & \bronze{400} & \bronze{7.7$\times$}
& 25.3$\pm$0.1 & 2034 & 5.1$\times$
\\

\hline
        
\hols
& \silver{46.1$\pm$0.1} & \bronze{1672} & \silver{2.0$\times$}
& \bronze{54.4$\pm$0.1} & 8552 & 9.9$\times$
& 34.1$\pm$0.3 & 566 & 10.9$\times$
& 23.6$\pm$0.0 & \bronze{510} & \bronze{1.3$\times$}
\\

        
\hline

\methodhom
& \bronze{45.6$\pm$0.1} & \gold{835} & \gold{1.0$\times$}
& \silver{56.9$\pm$0.2} & \silver{869} & \gold{1.0$\times$}
& \bronze{37.0$\pm$0.3} & \gold{52} & \gold{1.0$\times$}
& 24.3$\pm$0.0 & \silver{429} & \silver{1.1$\times$}
\\

\method
& \gold{80.4$\pm$0.0} & \silver{841} & \gold{1.0$\times$}
& \gold{67.3$\pm$0.1} & \gold{867} & \gold{1.0$\times$}
& \silver{38.9$\pm$0.1} & \gold{52} & \gold{1.0$\times$}
& \gold{28.7$\pm$0.1} & \gold{395} & \gold{1.0$\times$}
\\

\hline
\end{tabular}
}}
\end{table}

\begin{table}[!t]
\vspace{-6mm}
\caption{\emphasize{\method wins on Homophily datasets.} \label{tab:effecthom}}
\setlength\fboxsep{0pt}
\centering{\resizebox{0.97\textwidth}{!}{
\begin{tabular}{c | c r r | c r r | c r r | c r r }

\hline
{\bf Dataset} & \multicolumn{3}{c|}{\textbf{Facebook} \cite{rozemberczki2019multiscale}} & \multicolumn{3}{c|}{\textbf{GitHub} \cite{rozemberczki2019multiscale}} & \multicolumn{3}{c|}{\textbf{arXiv-Category} \cite{wang2020microsoft}} & \multicolumn{3}{c}{\textbf{Pokec-Locality} \cite{takac2012data}} \\
\hline
\# of Nodes / Edges / Classes 
& \multicolumn{3}{c|}{22.5K / 171K / 4} & \multicolumn{3}{c|}{37.7K / 289K / 2} 
& \multicolumn{3}{c|}{169K / 1.2M / 40} & \multicolumn{3}{c}{1.6M / 22.3M / 10} \\

Label Fraction
& \multicolumn{3}{c|}{4\%} & \multicolumn{3}{c|}{4\%}
& \multicolumn{3}{c|}{4\%} & \multicolumn{3}{c}{0.4\%}
\\

\hline

Method
& Accuracy & Time (s) & Rel. Time
& Accuracy & Time (s) & Rel. Time
& Accuracy & Time (s) & Rel. Time
& Accuracy & Time (s) & Rel. Time
\\

\hline

\gcn   
& 67.0$\pm$0.8 & \silver{12} & \silver{2.0$\times$}
& \silver{81.0$\pm$0.6} & \silver{28} & \silver{2.2$\times$}
& 25.4$\pm$0.3 & 216 & \bronze{2.3$\times$}
& 17.3$\pm$0.4 & \bronze{4002} & \silver{2.9$\times$}
\\

\appnp   
& 50.5$\pm$2.2 & \bronze{46} & \bronze{7.7$\times$}
& 74.2$\pm$0.0 & \bronze{73} & \bronze{5.6$\times$}
& 19.4$\pm$0.6 & 1176 & 12.3$\times$
& 16.8$\pm$1.7 & 11885 & 8.6$\times$
\\

        
\hline

\mixhop  
& \bronze{69.2$\pm$0.7} & 296 & 49.3$\times$
& 77.8$\pm$1.3 & 526 & 40.5$\times$
& 33.0$\pm$0.6 & 3203 & 33.4$\times$
& 16.9$\pm$0.3 & 52139 & 37.9$\times$
\\

\gprgnn
& 51.9$\pm$1.5 & 47 & 7.8$\times$
& 74.1$\pm$0.1 & 75 & 5.8$\times$
& 19.7$\pm$0.3 & 1174 & 12.2$\times$
& 30.0$\pm$2.0 & 11959 & 8.7$\times$
\\

\hline
        
\hols
& \gold{86.0$\pm$0.4} & 934 & 155.7$\times$
& \bronze{80.8$\pm$0.5} & 126 & 9.7$\times$
& \silver{61.4$\pm$0.2} & 627 & 6.5$\times$
& \bronze{63.7$\pm$0.3} & 8139 & \bronze{5.9$\times$}
\\

        
\hline

\methodhom
& \silver{85.2$\pm$0.5} & \gold{6} & \gold{1.0$\times$}
& \gold{81.3$\pm$0.5} & \gold{13} & \gold{1.0$\times$}
& \gold{61.7$\pm$0.2} & \gold{96} & \gold{1.0$\times$}
& \gold{66.0$\pm$0.2} & \silver{1437} & \gold{1.0$\times$}
\\

\method
& \silver{85.2$\pm$0.5} & \gold{6} & \gold{1.0$\times$}
& \gold{81.3$\pm$0.5} & \gold{13} & \gold{1.0$\times$}
& \bronze{58.8$\pm$0.6} & \silver{108} & \silver{1.1$\times$}
& \silver{64.8$\pm$0.8} & \gold{1377} & \gold{1.0$\times$}
\\

\hline
\end{tabular}
}}
\vspace{-3mm}
\end{table} 

\subsection{Q1 - Accuracy}
\vspace{-2mm}
In Table~\ref{tab:effecthet} and \ref{tab:effecthom}, we report the accuracy and running time. 
We highlight the top three from dark to light by \goldc{~}, \silverc{~} and \bronzec{~} denoting the first, second and third place.
\textit{In summary, \method wins on \xophily, heterophily and homophily graphs.}\looseness=-1




\noindent\textbf{\xophily and Heterophily.}
In Table~\ref{tab:effecthet}, \method outperforms all the competitors significantly by more than $34.3\%$ and $12.9\%$ accuracy on ``Synthetic'' and ``Pokec-Gender'', respectively.
These graphs exhibit strong \nef, thus \method boosts the accuracy owing to precise estimations of compatibility matrix. 
Heterophily GNNs give results close to majority voting when the observed labels are not adequate. 
With homophily assumption, General GNNs and BP-based methods also not perform well.
Both ``arXiv-Year'' and ``Patent-Year'' have weak \nef (Sec.~\ref{ssec:discover}).
Even so, \method still outperforms the competitors by estimating a reasonable compatibility matrix (Fig.~\ref{fig:ex4}).



\noindent\textbf{Homophily.}
In Table~\ref{tab:effecthom}, \methodhom outperforms all the competitors on $3$ out of $4$ homophily graphs, namely ``GitHub'', ``arXiv-Category'' and ``Pokec-Locality'', and \method performs similarly to \methodhom.
In addition, \methodhom performs competitively with \hols on ``Facebook'', while being $155.7\times$ faster.

\noindent\textbf{Ablation Study.}
\textit{Our optimizations make a difference.}
We evaluate different compatibility matrices -- 
(i) \method-EC uses edge counting on the labels of adjacent nodes in the priors, and
(ii) \method-A uses the adjacency matrix instead of \emphasis as the input of \methodest.
To evaluate the cases when imbalanced labels are given, we upsample $5\%$ labels to the class with the fewest labels in the datasets with weak \nef during the estimation.
In Table~\ref{tab:ablation}, we find that \method outperforms all its variants in all datasets.
In the graphs with strong \nef, \method shows its robustness to the structural noises and gives better results.
In the imbalanced graphs, while \method-EC brings its vulnerability to light, \method stays with high accuracy.
This study highlights the importance of a compatibility matrix estimation, as well as forming it into an optimization problem as shown in Lemma~\ref{lem:nef}. 

\subsection{Q2 - Scalability}
\vspace{-2mm}
\textit{\method is scalable and thrifty.}
We vary the edge number in ``Pokec-Gender'' and plot against the running time, including training and inference.
In Fig.~\ref{fig:scale}, \method scales linearly as expected (Lemma~\ref{lem:complexity}).
Table~\ref{tab:dollar} shows the estimated AWS dollar cost in ``Pokec-Gender'', assuming that we use a CPU machine for \method, and a GPU one for \gcn.
\method achieves up to $45\times$ savings. 
Details in Appx.~\ref{ssec:dollar}.

\vspace{-5mm}
\subsection{Q3 - Explainability} \label{ssec:qtexplain}
\vspace{-2mm}

Fig.~\ref{fig:ex} shows the compatibility matrices that \method recovered.
\textit{In a nutshell, the results agree well with the intuition.}
\hide{
We illustrate that the estimations of compatibility matrix by \method are reasonable in Fig.~\ref{fig:ex}, so as to interpreting the interrelations of classes extremely well.
The interrelations of shown estimated compatibility matrices are similar to the ones of edge counting in Fig.~\ref{fig:dis}, while being more robust to the noisy neighbors, namely, weakly connected ones.
}
For ``Synthetic'', \method matches the answer used for graph generation.
For ``Pokec-Gender'' (Fig.~\ref{fig:ex2}), \method report heterophily,
where people incline to have more opposite gender interactions \cite{ghosh2019quantifying}.
\hide{
Although ``arXiv-Year'' and ``Patent-Year'' do not have strong \nef, \method still gives an estimated compatibility matrices making much sense in the real-world (Fig.~\ref{fig:ex3} and~\ref{fig:ex4}), where the papers and patents tend to cite to the ones whose published dates are relatively close to them.
}
For ``arXiv-Year'' and ``Patent-Year'',
\method find that papers and patents tend to cite the ones published in nearby years, which also agrees with intuition (Fig.~\ref{fig:ex3} and~\ref{fig:ex4}).

\begin{table}[t]
\begin{minipage} {0.61\linewidth}
\caption{\emphasize{Ablation Study}: Estimating compatibility matrix by proposed \emphasis is essential. \label{tab:ablation}}
\vspace{-4mm}
\setlength\fboxsep{0pt}
\centering{\resizebox{1\columnwidth}{!}{
\begin{tabular}{c | C{1.2cm} | cccc}
\hline
\textbf{Datasets} & \nef Strength & \methodhom & \method-EC & \method-A & \method \\
\hline
\textbf{Synthetic} & \multirow{2}{*}{Strong} & 77.7$\pm$0.0 & 68.0$\pm$0.1 & 77.4$\pm$0.0 & \gold{80.5$\pm$0.0} \\
\textbf{Pokec-Gender} &  & 56.9$\pm$0.1 & 64.9$\pm$0.2 & 64.8$\pm$0.2 & \gold{67.3$\pm$0.1} \\
\hline
\textbf{arXiv-Year} (imba.) & \multirow{2}{*}{Weak} & 37.0$\pm$0.3 & 36.5$\pm$1.0 & 35.7$\pm$0.6 & \gold{38.4$\pm$0.0} \\
\textbf{Patent-Year} (imba.) &  & 24.1$\pm$0.0 & 24.0$\pm$0.9 & \gold{28.7$\pm$0.1} & \gold{28.7$\pm$0.0} \\
\hline
\end{tabular}
}}
\end{minipage} \hfill
\setlength{\tabcolsep}{1pt}
\begin{minipage} {0.38\linewidth}
\caption{\emphasize{\method is thrifty.} AWS dollar cost (\$) is reported, by \blue{t3.small} and \red{p3.2xlarge}. \label{tab:dollar}}
\setlength\fboxsep{0pt}
\centering{\resizebox{1\columnwidth}{!}{
\begin{tabular}{c | cc}
\hline
\textbf{Datasets} & \method & \gcn \\
\hline
\textbf{Pokec-Gender}   & \gold{\blue{\$ 0.33 (1.0$\times$)}} & \red{\$ 12.61 (45.0$\times$)} \\
\textbf{Pokec-Locality} & \gold{\blue{\$ 0.53 (1.0$\times$)}} & \red{\$ 13.66 (29.1$\times$)} \\
\hline
\end{tabular}
}}
\end{minipage}
\vspace{-4mm}
\end{table}

\begin{figure}[t]
\begin{minipage} {0.38\linewidth}
\centering
\includegraphics[scale=0.35]{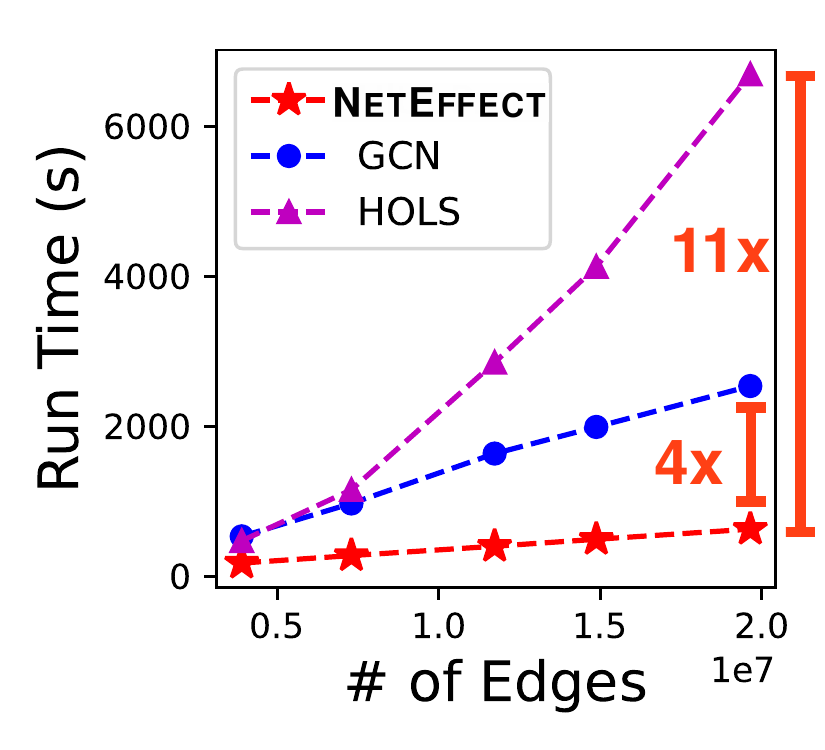}
\vspace{-3mm}
\caption{\emphasize{\method is scalable.} It is fast and scales linearly with the edge number.} \label{fig:scale}
\end{minipage} \hfill
\begin{minipage} {0.6\linewidth}
\centering
\subfloat[\label{fig:ex2} \scriptsize ``Pokec-Gender'':\\Heterophily]
{\includegraphics[height=0.9in]{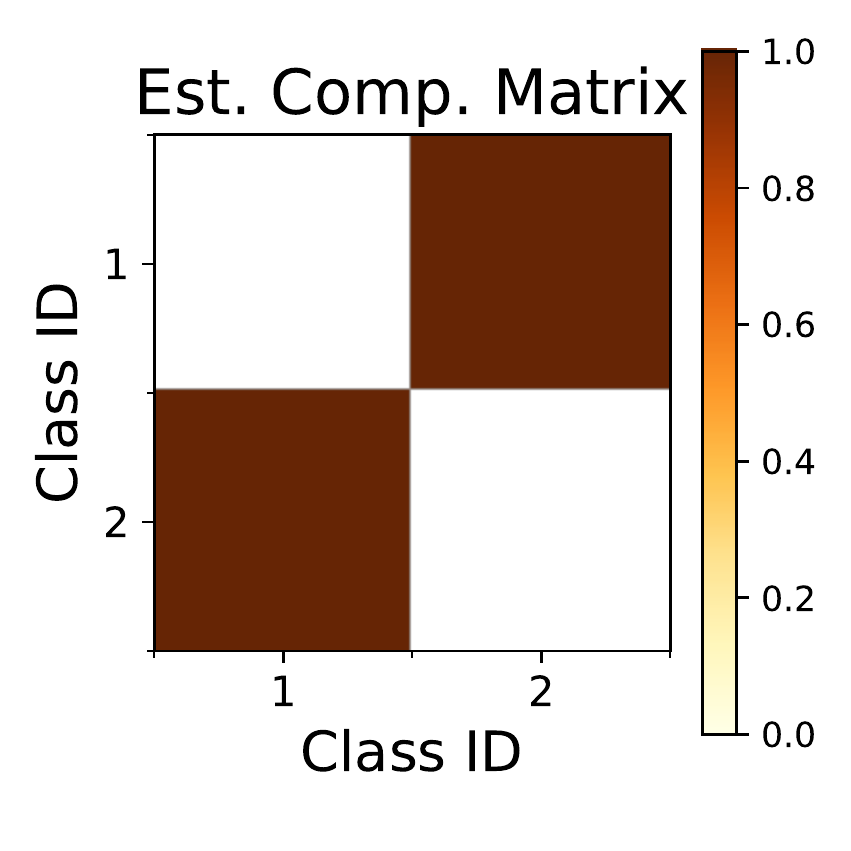}}
\subfloat[\label{fig:ex3} \scriptsize ``arXiv-Year'':\\\xophily]
{\includegraphics[height=0.9in]{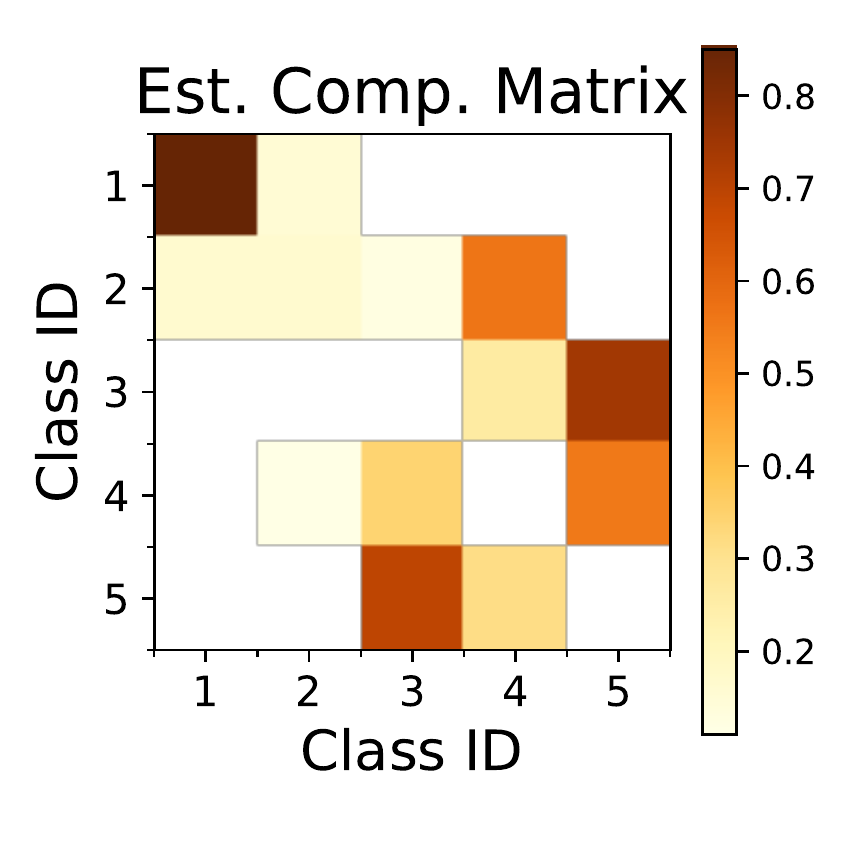}}
\subfloat[\label{fig:ex4} \scriptsize ``Patent-Year'':\\Heterophily]
{\includegraphics[height=0.9in]{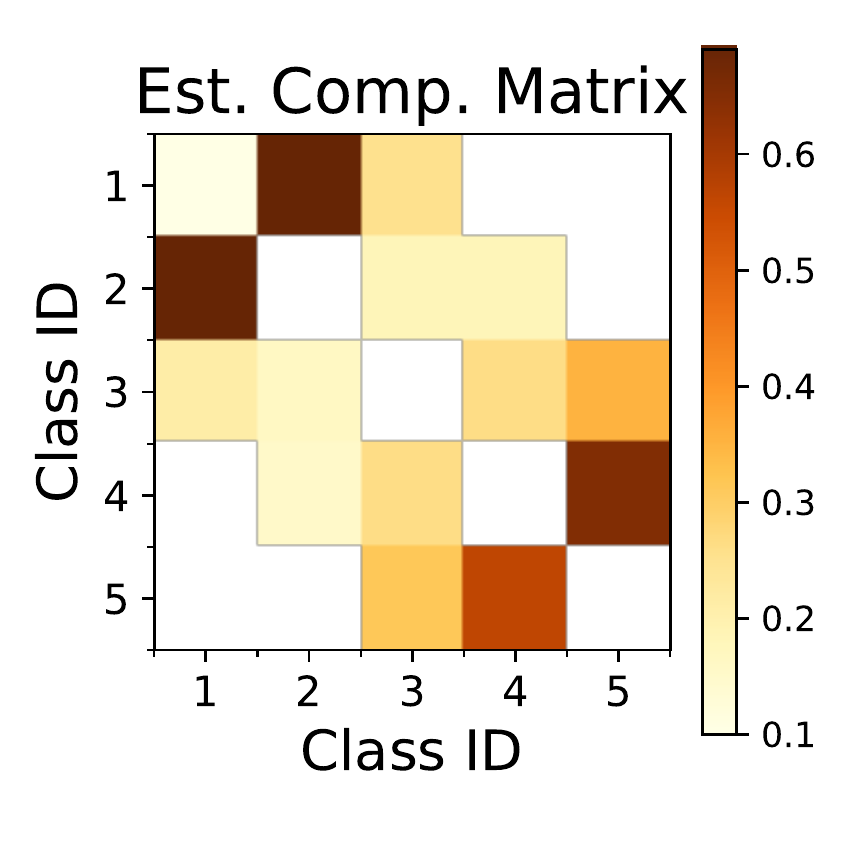}}
\vspace{-3mm}
\caption{\label{fig:ex} \emphasize{\method is explainable.} Our estimated compatibility matrices are much more robust to noises compared to edge counting (in Fig.~\ref{fig:dis}).}
\end{minipage}
\vspace{-7mm}
\end{figure}


\vspace{-5mm}
\section{Conclusions} \label{sec:concl}
\vspace{-2mm}


 We analyze the \neteffect (\nef) in node classification in the presence of only few labels.
Our proposed \method has the following  desirable properties:
\ben
\item {\em \theory}: \methodtest to statistically identify the presence of \nef,
\item {\em \general} and {\em \explain}: \methodest to estimate \nef with derived closed-form solution, if there is any, and
\item {\em \accurate} and {\em \scale}: \methodexp to efficiently exploit \nef for better performance on node classification.
\een
Applied on a real-world graph with {\em$22.3$M} edges, \method only requires {\em $14$ minutes}, and outperforms baselines on both accuracy and speed ($\geq 7\times$).



\vspace{-3mm}
\bibliographystyle{splncs04}
\bibliography{BIB/ref}

\newpage
\appendixpageoff
\appendixtitleoff
\renewcommand{\appendixtocname}{Supplementary Material}
\begin{appendix}
\label{sec:app}




\section{Proof}

\subsection{Proof of Lemma~\ref{lem:nef}}
\label{ap:subsec:proofprop1}

\begin{proof}
We derive the solution by the closed formula of linear regression and use the property of the mixed Kronecker matrix-vector product.
By vectorizing Eqn.~\ref{eq:simple}, we have:
\begin{align}
\vspace{-1mm}
\text{vec}{(\hat{{\boldsymbol E}})} & = \text{vec}{({\boldsymbol A}\hat{{\boldsymbol E}}\hat{{\boldsymbol H}})} \\
& = \text{vec}{(({\boldsymbol A}\hat{{\boldsymbol E}})\hat{{\boldsymbol H}}{\boldsymbol I}_{c \times c})} \\
& = ({\boldsymbol I}_{c \times c} \otimes ({\boldsymbol A}\hat{{\boldsymbol E}}))\text{vec}{(\hat{{\boldsymbol H}})}
\vspace{-1mm}
\end{align}
Let ${\boldsymbol X} = {\boldsymbol I}_{c \times c} \otimes ({\boldsymbol A}\hat{{\boldsymbol E}})$ and ${\boldsymbol y} = \text{vec}{(\hat{{\boldsymbol E}})}$, the above equation can be solved by linear regression with the closed-form solution of $\text{vec}{(\hat{{\boldsymbol H}})}$,
which completes the proof.
\hfill$\blacksquare$
\end{proof}

\subsection{Proof of Lemma~\ref{lem:crw1} and~\ref{lem:crw2}}
\label{ap:subsec:prooflemma12}
\begin{proof}
For a $L$-steps random walk sequence $S$ with $M$ trials, its length $|S|$ is $LM$. A random variable $X$ denotes the probability of node $i$ will walk to its $j$-th neighbor is:
\begin{equation}
    X = \mathbb{P}(\text{node } i \text{ walks to } N(i)_{j}) = \frac{\sum_{k=1}^{|S|}{\mathbbm{1}(N(i)_{j} = S_{k})}}{|S|},
\end{equation}
where $\mathbb{P}$ is the probability and $\mathbbm{1}$ is the indicator. With regular random walk in the graph, $X$ is upper-bounded by $\frac{\lceil (L - 1) / 2 \rceil}{L}$. By applying Hoeffding's inequality, we have:
\begin{equation}
    \mathbb{P}(|\hat{\mu}_{|S|} - \mu| \geq \epsilon) \leq 2\exp{\frac{-2L^{3}Mt^{2}}{\lceil (L - 1) / 2 \rceil^{2}}},
\end{equation}
where $\hat{\mu}_{|S|}$ denotes the sampled mean of the given random variable, and $\mu$ denotes the expectation. Let $\delta = 2\exp{\frac{-2L^{3}Mt^{2}}{\lceil (L - 1) / 2 \rceil^{2}}}$, with probability $1 - \delta$, we have the error:
\begin{equation}
    \epsilon = |\hat{\mu}_{|S|} - \mu| \leq \frac{\lceil (L - 1) / 2 \rceil}{L} \sqrt{\frac{\log{(2/\delta)}}{2LM}}
\end{equation}
With non-backtracking random walk \cite{alon2007non}, the upper bound of $X$ can be decreased to $\frac{\lceil (L - 1) / 3 \rceil}{L}$. Let $\delta = 2\exp{\frac{-2L^{3}Mt^{2}}{\lceil (L - 1) / 3 \rceil^{2}}}$, with probability $1 - \delta$, we now have the error:
\begin{equation}
    \epsilon = |\hat{\mu}_{|S|} - \mu| \leq \frac{\lceil (L - 1) / 3 \rceil}{L} \sqrt{\frac{\log{(2/\delta)}}{2LM}}
\end{equation}
\hfill$\blacksquare$
\end{proof}


\subsection{Proof of Lemma~\ref{lem:con}} \label{ap:subsec:proofcon}
\begin{proof}
\method exactly converges if and only if $\rho{({\boldsymbol A}^{*})}\rho{(\hat{\boldsymbol H}^{*})} < 1$. 
${\boldsymbol H}^{*}$ is row-normalized, where $\rho{({\boldsymbol H}^{*})} = 1$ is a constant, and is less than $1$ after centering.
The scaling factor $f$ multiplied to the propagation has to be in the range of $(0, 1/\rho{({\boldsymbol A}^{*})})$ to meet the criterion of exact convergence.
\end{proof}

\subsection{Proof of Lemma~\ref{lem:complexity}} \label{ap:subsec:proofcomplexity}
\vspace{-2mm}
\begin{proof}
For \methodest, since there are $c$ sets of parameters are independent, we can separate the problem into $c$ tasks, where each contains $c$ features and $|\mathcal{P}|$ samples. 
The complexity can then be reduced to $O(|\mathcal{P}| \cdot c^{3})$, and the efficient leave-one-out cross-validation only needs to be done once.
For \methodexp, for each random walk, each node visits at most $L \cdot M$ unique nodes, so the maximum number of non-zero elements in ${\boldsymbol W}$ is either $n \cdot L \cdot M$ if we have not walked through all the edges, or $m$ otherwise. 
SVD on ${\boldsymbol W}$ takes $O(d \cdot \max{(m, n \cdot L \cdot M)})$. 
It takes at most $O(m + n)$ for sparse matrix multiplication to run $t$ iterations.
Thus, the time complexity is $O(d \max{(m, n \cdot L \cdot M)} + |\mathcal{P}| \cdot c^{3} + m)$.
In practice, $c$, $|\mathcal{P}|$ and $t$ are usually small constants which are negligible, and $m$ is usually much larger. 
Keeping only the dominating terms, the time complexity is approximately $O(m)$.
${\boldsymbol W}$ contains at most $\max{(m, n \cdot L \cdot M)}$ non-zero elements. The Kronecker product at most contains $n \cdot c^{2}$ non-zero elements. 
The space complexity is $O(\max{(m, n \cdot L \cdot M)} + n \cdot c^{2})$.
\vspace{-4mm}
\end{proof}



\section{Algorithms}
\vspace{-3mm}
\subsection{\methodtest} \label{algo:fisher}
\vspace{-2mm}
\begin{algorithm}[H]
\KwData{Edges $\mathcal{E}$ and priors $\mathcal{P}$}
\KwResult{$p$-value table ${\boldsymbol F}$}
\tcc{edges with both nodes in priors}
Extract $\mathcal{E}^{'}$ such that $(i, j) \in \mathcal{E}, i, j \in \mathcal{P} \ \forall (i, j) \in \mathcal{E}^{'}$\;
${\boldsymbol T} \leftarrow {\boldsymbol O}_{c \times c}$\tcp*{test statistic table}
\tcc{do $\chi^{2}$ test for $B$ times ($B = 1000$ by default)}
\For{$b_{1}=1,...,B$}{
    \For{$c_{1}=1,...,c-1$}{
        \For{$c_{2}=c_{1}+1,...,c$}{
            ${\boldsymbol V} \leftarrow {\boldsymbol O}_{2 \times 2}$\tcp*{contingency table}
            \For{$(i, j) \in \text{Sampled}(\mathcal{E}^{'})$}{
                \uIf{$l(i) = c_{1}$ and $l(j) = c_{1}$}{
                    ${\boldsymbol V}_{11} \leftarrow {\boldsymbol V}_{11} + 2$\;
                }
                \uElseIf{($l(i) = c_{1}$ and $l(j) = c_{2}$) or \\($l(i) = c_{2}$ and $l(j) = c_{1}$)}{
                    ${\boldsymbol V}_{21} \leftarrow {\boldsymbol V}_{21} + 1$ and ${\boldsymbol V}_{12} \leftarrow {\boldsymbol V}_{12} + 1$\;
                }
                \uElseIf{$l(i) = c_{2}$ and $l(j) = c_{2}$}{
                    ${\boldsymbol V}_{22} \leftarrow {\boldsymbol V}_{22} + 2$\;
                }
                \If{$\sum_{i=1}^{2}{\sum_{j=1}^{2}{\boldsymbol {\boldsymbol V}_{ij}}} > 500$}{
                    Break\;
                }
            }
            \tcc{record statistics of class pairs}
            $T = \chi^{2}\text{-Test-Statistic}({\boldsymbol V} / 2)$\;
            ${\boldsymbol T}_{c_{1}c_{2}} \leftarrow {\boldsymbol T}_{c_{1}c_{2}} + T / B $ and ${\boldsymbol T}_{c_{2}c_{1}} \leftarrow {\boldsymbol T}_{c_{2}c_{1}} + T / B$\;
        }
    }
}
Compute $p$-value table ${\boldsymbol F}_{c \times c}$ with average statistics in ${\boldsymbol T}$\;
Return ${\boldsymbol F}$\;
\caption{\methodtest}
\end{algorithm}

\subsection{\methodexp} \label{algo:main}
\begin{algorithm}[H]
\KwData{``Emphasis'' matrix ${\boldsymbol A}^{*}$, estimated compatibility matrix $\hat{\boldsymbol H}^{*}$, and initial belief $\hat{{\boldsymbol E}}$}
\KwResult{Final belief ${\boldsymbol B}$}
$\hat{{\boldsymbol B}}_{(0)} \leftarrow {\boldsymbol O}_{n \times c}, t \leftarrow 0$\;
\While{$\lVert\hat{{\boldsymbol B}}_{(t+1)} - \hat{{\boldsymbol B}}_{(t)}\rVert_{1} > 1$}{
   $\hat{{\boldsymbol B}}_{(t+1)} \leftarrow \hat{{\boldsymbol E}} + f {\boldsymbol A}^{*}\hat{{\boldsymbol B}}_{(t)}\hat{{\boldsymbol H}}^{*}$\;
   $t \leftarrow t + 1$\;
}
Return ${\boldsymbol B} \leftarrow \hat{{\boldsymbol B}}_{(t)} + 1/c$\;
\caption{\methodexp}
\end{algorithm}

\section{Reproducibility} \label{sec:rep}

\subsection{Datasets}
We include $3$ citation networks: ``arXiv-Year'' \cite{hu2020open}, ``Patent-Year'' \cite{leskovec2005graphs}, and ``arXiv-Category'' \cite{wang2020microsoft}, 
and $4$ social networks: ``Pokec-Gender'' \cite{takac2012data}, ``Facebook'' \cite{rozemberczki2019multiscale}, ``GitHub'' \cite{rozemberczki2019multiscale}, and ``Pokec-Locality'' \cite{takac2012data}.
``Synthetic'' is the enlarged graph in Fig.~\ref{fig:c2}. 
Noisy edges are injected, and the communities are connected by dense blocks.

\setlength\intextsep{0pt}
\begin{wrapfigure}{L}{0.5\columnwidth}
\begin{minipage}[t]{0.5\columnwidth}
\captionof{table}{Hyperparameters for GNNs \label{tab:hyper}}
\centering{\resizebox{1\textwidth}{!}{
\begin{tabular}{c | l}
\hline
Method & Hyperparameters \\
\hline
\gcn & lr=0.01, wd=0.0005, hidden=16, dropout=0.5 \\ 
\appnp & lr=0.002, wd=0.0005, hidden=64, dropout=0.5, K=10, alpha=0.1 \\
\mixhop & lr=0.01, wd=0.0005, cutoff=0.1, layers1=[200, 200, 200], layers2=[200, 200, 200] \\
\gprgnn & lr=0.002, wd=0.0005, hidden=64, dropout=0.5, K=10, alpha=0.1 \\
\hline
\end{tabular}
}}
\end{minipage}
\end{wrapfigure}

\subsection{Hyperparameters}
For \method, we use random walks of length $4$ with $10$ trials except ``GitHub'', ``arXiv-Category'' and ``Pokec-Locality'', where we use $30$ trials. 
The decomposition rank is set to be $256$.
The weights of \hols for different motifs are set to be equal.
For GNNs, we train for $200$ epochs and use the default hyperparameters given by the authors in Table~\ref{tab:hyper}.





\subsection{Scalability} \label{ssec:dollar}
We select AWS machines with comparable specs.
For CPU machine, we select t3.small with 3.3GHz CPU, which is faster than ours, and costs $\$0.023$ per hour.
For GPU machine, we select p3.2xlarge with a V100 GPU, which costs $\$3.06$ per hour, which is 0.89 slower than the RTX A6000 GPU we use on running PyTorch\footnote{https://lambdalabs.com/blog/nvidia-rtx-a6000-benchmarks/}.
The running time of \gcn on ``Pokec-Gender'' and ``Pokec-Locality'' are $673$ and $730$ seconds, respectively.
Using the provided information, the results in Table~\ref{tab:dollar} can be computed.
\end{appendix}



\end{document}